\pdfoutput=1
\documentclass[journal]{IEEEtran}
\usepackage{amsmath,amsfonts,amssymb}
\usepackage{graphicx}
\usepackage{cite}
\usepackage{multirow}
\usepackage{algorithm}
\usepackage{algorithmic}
\usepackage{subfigure}
\usepackage{url}

\IEEEoverridecommandlockouts

\newtheorem{theorem}{Theorem}

\newtheorem{lemma}{Lemma}
\newtheorem{corollary}{Corollary}

\newtheorem{example}{Example}

\newtheorem{fact}{Fact}

\begin{document}
\title{On the Capacity Region of Two-User Linear Deterministic Interference Channel and Its Application to Multi-Session Network Coding}
\author{Xiaoli Xu, Yong Zeng, Yong Liang Guan and Tracey, Ho
\thanks{X. Xu, Y. Zeng and Y. L. Guan are with the School of Electrical and Electronic Engineering, Nanyang Technological University, Singapore 639801
(email: \{xu0002li, ze0003ng, eylguan\}@e.ntu.edu.sg)}
\thanks{T. Ho is with the Department of Electrical Engineering, California Institute of Technology, Pasadena, California 91125, USA (email:tho@caltech.edu)}
\thanks{This work was supported by the Advanced Communications Research Program DSOCL06271, a research grant from the Directorate of Research and Technology (DRTech), Ministry of Defence, Singapore.}}
\maketitle
\begin{abstract}

In this paper, we study the capacity of the two-user multiple-input multiple-output (MIMO) linear deterministic interference channel (IC), with possible correlations within/between the channel matrices. The capacity region is characterized in terms of the rank of the channel matrices. It is shown that \emph{linear precoding} with Han-Kobayashi type of rate-splitting, i.e., splitting the information-bearing symbols of each user into common and private parts, is sufficient to achieve all the rate pairs in the derived capacity region. The capacity result is applied to obtain an achievable rate region for the double-unicast networks with random network coding at the intermediate nodes, which can be modeled by the two-user MIMO linear deterministic IC studied. It is shown that the newly proposed achievable region is strictly larger than the existing regions in the literature.

\end{abstract}

\section{Introduction}\label{sec:intro}
The two-user interference channel (IC) models the communication between two source-receiver pairs via a common channel. As there is no cooperation between any of the sources and receivers, the transmission from one source to its corresponding receiver is viewed as interference by the other source-receiver pair.  The capacity for the general two-user IC is a long-term fundamental open problem since first studied by Shannon in \cite{Shannon61}. The best achievable rate region to date is established by Han and Kobayashi \cite{Han81}, where a common-private rate splitting technique is employed to enable the receiver to partially decode and subtract the interference. It was later pointed out in \cite{Gamal82} that the Han-Kobayashi region matches the capacity region for a class of deterministic ICs. Recent breakthroughs in studying the capacity of two-user IC shows that a simple linear deterministic model captures the key properties of the Gaussian channel \cite{Avestimehr11} and it leads to capacity characterization within a constant number of bits \cite{Bresler08}. Moreover, the linear deterministic channel model is closely related with the degrees-of-freedom (DoF) characterization of the two-user Gaussian IC as both of them focus on the high signal-to-noise ratio regime where noise is de-emphasized in order to get a better understanding of the interference \cite{Jafar09}. In fact, as pointed out in \cite{Bresler08}, the capacity region of the deterministic channel is, when properly scaled, equal to the generalized DoF region.

The DoF characterization of the two-user MIMO Gaussian IC is given in \cite{Jafar07} with the assumption that the channels are nondegenerated, i.e. all channel matrices are full rank and independent of each other. Moreover, it is shown that zero forcing (ZF), which is normally a suboptimal strategy, is sufficient to achieve all DoF. Unfortunately, the result obtained in \cite{Jafar07} is no longer applicable if the channel matrices are correlated and/or rank deficient. In this paper, we generalize the results in \cite{Jafar07} by removing the assumptions on the channel matrices. Specifically, we study the capacity of the two-user MIMO deterministic IC, i.e., the IC studied in \cite{Jafar07} but with the additive noise term set to zero. In contrast to \cite{Jafar07} where all channels are assumed to be of full rank and independent, we consider the more general case that the channel matrices may be rank deficient and/or correlated with each other. This channel model is of theoretical interest by itself and also renders its application in obtaining an achievable rate region for the double-unicast networks as discussed later. The exact capacity region (or equivalently the DoF region if the noise term is non-zero) is characterized in terms of the rank of the channel matrices. The capacity achieving scheme is given by \emph{linear precoding} together with Han-Kobayashi type of rate-splitting, i.e., the data symbols are split into a common part, which is decodable at both receivers, and a private part, which is decodable at the intended receiver only. Furthermore, the precoder consists of a random spreading matrix, which maps the data symbols into a subspace of higher dimension, and a ZF precoding matrix, which effectively pre-cancel the private symbols. Note that the capacity achieving scheme is quite simple and efficient due to its linear precoding and decoding processes.


As an application for the capacity results derived for the two-user deterministic IC, we obtain a strictly enlarged achievable rate region  as compared to existing schemes for the double-unicast networks. Both the two-user ICs and double-unicast networks share the similarity that each source is intended to send an independent message to its corresponding destination and it causes interference to the other source-destination pair. However, different from wireless ICs where the channels are determined by nature (e.g., channel gain, fading and so on), in double-unicast networks, the signals from the source pass through a set of intermediate nodes before arriving at the destination. As a result, the processing strategies used by the intermediate nodes directly affect the achievable rates of the network. It was recently shown that by employing network coding at the intermediate nodes, the throughput can be significantly improved as
compared to traditional techniques such as routing \cite{Erez09}. For single-session networks, it has been shown that
random network coding is capacity achieving \cite{Ho06}. However, for double-unicast networks, the optimal network coding strategy remains unknown. In this paper, with all the intermediate nodes performing random linear network coding, we show that the input-output relationship of the double-unicast networks can be modeled by the two-user linear deterministic IC. As such, the capacity results we derived can be directly applied to obtain an achievable region.

There are some existing works on the achievable rate region characterizations for the double-unicast networks \cite{Traskov06,Erez09,Huang11}.
In \cite{Traskov06}, the problem was formulated as a linear programming problem by packing butterfly structures in the network. However, this approach
is limited since only XOR operation is allowed in the butterfly structure. In \cite{Erez09}, an achievable region was obtained by using the so-called
``rate-exchange'' method, where starting from the single-user rate for one of the users, a non-zero rate for the other user is achieved by directly sacrificing the
single-user rate via some interference nulling schemes. In \cite{Huang11}, another rate region was obtained by using some precoding techniques.
As stated in \cite{Huang11}, the region obtained in \cite{Erez09} does not contain that given in \cite{Huang11}, and vice versa. In this paper,
we will show analytically that our proposed region by utilizing the linear deterministic IC model contains both that in \cite{Erez09} and \cite{Huang11}.

The rest of this paper is organized as follows. Section \ref{sec:prob} introduces the system model. The main results are presented in Section \ref{sec:main}, which characterizes the capacity region of the two-user linear deterministic IC. In Section \ref{sec:nc}, the obtained capacity results are utilized to derive an achievable region for the double-unicast networks and the comparison with existing works is given. Finally, this paper is concluded in Section \ref{sec:conclusion}.

\emph{Notations}: Throughout this paper, $\mathbb{R}^{n\times m}$ denotes the space of  $n \times m$ real matrices and $\mathbb{F}_q^{n\times m}$ represents the space of $n\times m$ matrices in a finite field of order $q$. Vectors are represented by boldface lower-case letters, e.g. $\mathbf{v}$.  Matrices are denoted by italicized capital letters, e.g. $A$. $A^T$ denotes the transpose of $A$. $0_{n\times m}$ represents a zero matrix of size $n\times m$ and the subscripts are omitted when there is no ambiguity from the context. The range (or column space) and null space of a matrix $A$ are denoted by $\mathcal{R}(A)$ and $\mathcal{N}(A)$, respectively.
\section{System Model}\label{sec:prob}
Consider the two-user linear deterministic IC as shown in Fig.\ref{F:DIC}, with sources denoted by $s_1,s_2$ and  destinations by $t_1,t_2$, respectively. The input-output relationships are given by\footnote{The result can be extended to complex-valued channels as well}
\begin{equation}
\begin{aligned}\label{eq:DIC}
\mathbf{y}_1&=H_{11}\mathbf{x}_1+H_{12}\mathbf{x}_2\\
\mathbf{y}_2&=H_{21}\mathbf{x}_1+H_{22}\mathbf{x}_2
\end{aligned}
\end{equation}
where $\mathbf{x}_1\in\mathbb{R}^{m_1}$, $\mathbf{x}_2\in\mathbb{R}^{m_2}$ represent the independent input vectors  by $s_1$ and $s_2$, respectively; $H_{ij}\in \mathbb{R}^{n_i \times m_j},i,j\in\{1,2\}$ denotes the channel matrix from $s_i$ to $t_j$; and  $\mathbf{y}_i\in\mathbb{R}^{n_i}$ is the channel output at the $i^{th}$ receiver who is interested in recovering $\mathbf{x}_i$ only.

\begin{figure}[htb]
\centering
\includegraphics[scale=0.8]{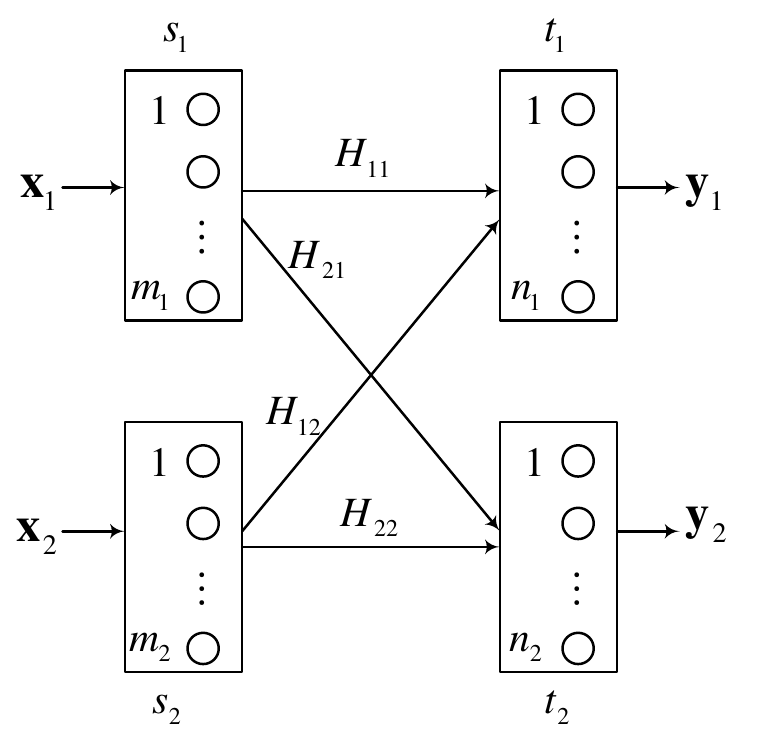}
\caption{Two-user MIMO linear deterministic IC}
\label{F:DIC}
\end{figure}

We are interested in determining the number of independent symbols that can be simultaneously and reliably transmitted from $s_1$ to $t_1$ and  $s_2$ to $t_2$, denoted as $R_1$ and $R_2$, respectively. The point-to-point capacity of the channel from  $s_i$ to $t_i$ is determined by the rank of the channel matrix, i.e., $R_i=\mathrm{rank}(H_{ii})$, $i\in\{1,2\}$. Note that when all the channel matrices are randomly generated (thus full rank) and independent of each other, the region of $R_1, R_2$ is equivalent to the DoF region of the two-user Gaussian IC and is characterized by \cite{Jafar07}:
\begin{equation}\label{eq:JafarBound}
\begin{aligned}
R_1\leq \min&\{m_1,n_1\}\\
R_2\leq \min&\{m_2,n_2\}\\
R_1+R_2\leq \min&\{m_1+m_2,n_1+n_2,\\
&\max(m_1,n_2),\max(m_2,n_1)\}
\end{aligned}
\end{equation}

However, for general $H_{11}$, $H_{12}$, $H_{21}$ and $H_{22}$ that may be correlated and/or rank deficient, the region given by \eqref{eq:JafarBound} is no longer applicable. Therefore, a more general result than \eqref{eq:JafarBound} is necessary.

Before presenting the main results, the following assumptions are made without loss of generality:
\begin{align}
\mathrm{rank}[\begin{matrix}H_{11}& H_{12}\end{matrix}]&=n_1\label{eq:asm1}\\
\mathrm{rank}[\begin{matrix}H_{21}& H_{22}\end{matrix}]&=n_2\label{eq:asm2}\\
\mathrm{rank}\bigg[\begin{matrix}H_{11} \\ H_{21}\end{matrix}\bigg]&=m_1\label{eq:asm3}\\ \mathrm{rank}\bigg[\begin{matrix}H_{12}\\H_{22}\end{matrix}\bigg]&=m_2\label{eq:asm4}
\end{align}

The above assumptions can be validated by showing that the capacity region will be unaffected by removing  the dependent received or transmitted symbols. For example, if the matrix $[\begin{matrix}H_{11}& H_{12}\end{matrix}]$ does not have full row rank, certain symbols received by $t_1$ are linear combinations of the rest and therefore they can be discarded without losing any information.

\section{Main Results}\label{sec:main}
\begin{theorem}\label{tho:main}
For the two-user MIMO linear deterministic IC given in \eqref{eq:DIC}, the capacity region is characterized by
\begin{align}
R_1\leq & \mathrm{rank}(H_{11})\label{eq:1}\\
R_2\leq & \mathrm{rank}(H_{22})\label{eq:2}\\
R_1+R_2\leq & n_1+m_2-\mathrm{rank}(H_{12})\label{eq:3}\\
R_1+R_2\leq & n_2+m_1-\mathrm{rank}(H_{21})\label{eq:4}\\
R_1+R_2\leq & \mathrm{rank}\bigg[\begin{matrix}H_{11} & H_{12} \\ H_{21} & 0_{n_2\times m_2}\end{matrix}\bigg] + \mathrm{rank}\bigg[\begin{matrix}H_{21} & H_{22} \\ 0_{n_1\times m_1} & H_{12}\end{matrix}\bigg] \nonumber \\ &-\mathrm{rank}(H_{21})-\mathrm{rank}(H_{12})\label{eq:5}\\
2R_1+R_2\leq & n_1+m_1+\mathrm{rank}\bigg[\begin{matrix}H_{21} & H_{22} \\ 0_{n_1\times m_1} & H_{12}\end{matrix}\bigg]\nonumber \\
  &-\mathrm{rank}(H_{21})-\mathrm{rank}(H_{12})\label{eq:6}\\
R_1+2R_2\leq & n_2+m_2+\mathrm{rank}\bigg[\begin{matrix}H_{11} & H_{12} \\ H_{21} & 0_{n_2\times m_2}\end{matrix}\bigg]\nonumber \\
&-\mathrm{rank}(H_{21})-\mathrm{rank}(H_{12})\label{eq:7}
\end{align}
\end{theorem}

Before proceeding to the proof, we give an alternative presentation of the region specified in \eqref{eq:1}-\eqref{eq:7}.
Firstly, denote $\mathrm{rank}(H_{ij})$ by $r_{ij}$, $i,j\in\{1,2\}$ and let the singular value decomposition (SVD) of $H_{12}$, $H_{21}$ be expressed as $H_{12}=U_1\Lambda_1V_1^T$ and $H_{21}=U_2\Lambda_2V_2^T$, where $U_1,V_1,U_2$ and $V_2$ are $n_1\times n_1$, $m_2\times m_2$, $n_2\times n_2$ and $m_1\times m_1$ orthogonal matrices, respectively; $\Lambda_1,\Lambda_2$ are $n_1\times m_2$ and $n_2\times m_1$ diagonal matrices with  singular values of $H_{12}, H_{21}$ on the main diagonal. Furthermore, $\Lambda_1$, $U_1$ and $V_1$ can be decomposed as follows:
\begin{itemize}
\item{$\Lambda_1=\left[\begin{matrix}D_{12} & 0_{r_{12}\times(m_2-r_{12})}\\ 0_{(n_1-r_{12})\times r_{12}} & 0_{(n_1-r_{12})\times(m_2-r_{12})}\end{matrix}\right]$, where $D_{12}\in\mathbb{R}^{r_{12}\times r_{12}}$ is a diagonal matrix with the non-zero singular values of $H_{12}$ on its main diagonal.}
\item{$U_1=\left[\begin{matrix}U_{11} & U_{10}\end{matrix}\right]$, where $U_{11}\in\mathbb{R}^{n_1\times r_{12}}$ whose columns form an orthonormal basis for the subspace spanned by columns of $H_{12}$, i.e. $\mathcal{R}(U_{11})=\mathcal{R}(H_{12})$; $U_{10}\in\mathbb{R}^{n_1\times(n_1-r_{12})}$ spans the null space of $H_{12}^{T}$, i.e $\mathcal{R}(U_{10})=\mathcal{N}(H_{12}^T)$.}
\item{$V_1=\left[\begin{matrix}V_{11} & V_{10}\end{matrix}\right]$, where $V_{11}\in\mathbb{R}^{m_2\times r_{12}}$ whose columns form an orthonormal basis for the subspace spanned by rows of $H_{12}$, i.e. $\mathcal{R}(V_{11})=\mathcal{R}(H_{12}^T)$; $V_{10}\in\mathbb{R}^{m_2\times(m_2-r_{12})}$ spans the null space of $H_{12}$, i.e. $\mathcal{R}(V_{10})=\mathcal{N}(H_{12})$.}
\end{itemize}
Similarly, $\Lambda_2,  U_2$ and $V_2$ can be decomposed as $\Lambda_2=\left[\begin{matrix}D_{21} & 0_{r_{21}\times(m_1-r_{21})}\\ 0_{(n_2-r_{21})\times r_{21}} & 0_{(n_2-r_{21})\times(m_1-r_{21})}\end{matrix}\right]$, $U_2=\left[\begin{matrix}U_{21} & U_{20}\end{matrix}\right]$, $V_2=\left[\begin{matrix}V_{21} & V_{20}\end{matrix}\right]$.\\

\begin{lemma}\label{lem:complexityReduce1}
$\mathrm{rank}\left[\begin{matrix}H_{11} & H_{12} \\ H_{21} & 0_{n_2\times m_2}\end{matrix}\right]=\mathrm{rank}(U_{10}^{T}H_{11}V_{20})+r_{21}+r_{12}$
\end{lemma}
\begin{proof}
Since $U_1,U_2,V_1,V_2$ are non-singular, $\left[\begin{matrix}U_1^T & 0\\ 0 & U_2^T\end{matrix}\right]$ and $\left[\begin{matrix}V_2 & 0\\ 0 & V_1\end{matrix}\right]$ are non-singular as well. By using the fact that the matrix rank is unchanged with a left or right multiplication by a non-singular matrix, we have
\begin{equation*}
\begin{aligned}
&\mathrm{rank}\left[\begin{matrix}H_{11} & H_{12} \\ H_{21} & 0_{n_2\times m_2}\end{matrix}\right]\\
&=\mathrm{rank}\left(\left[\begin{matrix}U_1^T & 0\\ 0 & U_2^T\end{matrix}\right]
\left[\begin{matrix}H_{11} & H_{12} \\ H_{21} & 0\end{matrix}\right]
\left[\begin{matrix}V_2 & 0\\ 0 & V_1\end{matrix}\right]\right) \\
&=\mathrm{rank}\left[\begin{matrix}U_{1}^{T}H_{11}V_2 & \Lambda_1 \\ \Lambda_2 & 0\end{matrix}\right]\\
&=\mathrm{rank}\left[\begin{matrix} U_{11}^{T}H_{11}V_{21} & U_{11}^{T}H_{11}V_{20} & D_{12} & 0\\
U_{10}^{T}H_{11}V_{21} & U_{10}^{T}H_{11}V_{20} & 0 & 0 \\
D_{21} & 0 & 0 & 0\\
0 & 0 & 0 & 0
\end{matrix}\right]\\
&\overset{(a)}{=}\mathrm{rank}\left[\begin{matrix} 0 & 0 & D_{12} \\
0 & U_{10}^{T}H_{11}V_{20} & 0\\
D_{21} & 0 & 0
\end{matrix}\right]\\
&=\mathrm{rank}(U_{10}^{T}H_{11}V_{20})+r_{21}+r_{12}
\end{aligned}
\end{equation*}
where (a) follows by elementary row and elementary column operations with the fact that $D_{12}$ and $D_{21}$ are non-singular.
\end{proof}
Following similar arguments as in Lemma \ref{lem:complexityReduce1}, we have
\begin{lemma}\label{lem:complexityReduce2}
$\mathrm{rank}\left[\begin{matrix}H_{21} & H_{22} \\ 0_{n_1\times m_1} & H_{12}\end{matrix}\right]=\mathrm{rank}(U_{20}^{T}H_{22}V_{10})+r_{21}+r_{12}$
\end{lemma}

With the above results,  \eqref{eq:1}-\eqref{eq:7} can be equivalently expressed as
\begin{align}
R_1\leq &r_{11} \label{eq:m1}\\
R_2\leq &r_{22} \label{eq:m2}\\
R_1+R_2\leq &n_1+m_2-r_{12}\label{eq:m3}\\
R_1+R_2\leq &n_2+m_1-r_{21}\label{eq:m4}\\
R_1+R_2\leq &\mathrm{rank}(U_{10}^{T}H_{11}V_{20})+\mathrm{rank}(U_{20}^{T}H_{22}V_{10})\nonumber\\
  &+r_{21}+r_{12}\label{eq:m5}\\
2R_1+R_2\leq &n_1+m_1+\mathrm{rank}(U_{20}^{T}H_{22}V_{10})\label{eq:m6}\\
R_1+2R_2\leq &n_2+m_2+\mathrm{rank}(U_{10}^{T}H_{11}V_{20})\label{eq:m7}
\end{align}

\subsection{Achievability}
This subsection gives the achievability proof of the capacity region given in Theorem~\ref{tho:main}. Given a rate pair $(R_1,R_2)$ that satisfies the inequalities in \eqref{eq:m1}-\eqref{eq:m7}, we show that it is achievable by using linear precoding together with a specific type of Han-Kobayahsi  rate-splitting.

Firstly, with the SVD expressions for $H_{12}$ and $H_{21}$ given previously, and by absorbing the orthogonal matrices $V_i^{T}$ into the input vector $\mathbf{x}_j$ and multiplying the output vector $\mathbf{y}_i$ with $U_{i}^T$, $i,j\in\{1,2\},i\neq j$, the channel model given in \eqref{eq:DIC} can be equivalently written as
\begin{align}
\mathbf{y}_1'&=H_{11}'\mathbf{x}_1'+\Lambda_1\mathbf{x}_2'\label{eq:mch1}\\
\mathbf{y}_2'&=H_{22}'\mathbf{x}_2'+\Lambda_2\mathbf{x}_1'\label{eq:mch2}
\end{align}
where $\mathbf{y}_1'=U_1^{T}\mathbf{y}_1$, $\mathbf{y}_2'=U_2^{T}\mathbf{y}_2$, $\mathbf{x}_1'=V_2^{T}\mathbf{x}_1$, $\mathbf{x}_2'=V_1^{T}\mathbf{x}_2$, $H_{11}'=U_1^{T}H_{11}V_2$ and $H_{22}'=U_2^{T}H_{22}V_1$. 
The advantage of this equivalent channel model is that it results in diagonal interfering channel matrices, which is easier to deal with. Similar transformations have been used in \cite{Jafar07}. To find the input signal vectors $\mathbf x_1$ and $\mathbf x_2$, it is sufficient to determine $\mathbf x_1'$ and $\mathbf x_2'$ since they are related by the nonsingular transformations given by $\mathbf x_1=V_2 \mathbf x_1'$ and $\mathbf x_2 = V_1 \mathbf x_2'$.

Denote by $\mathbf{d}_1$ and $\mathbf{d}_2$ the information-bearing symbols to be sent by $s_1$ and $s_2$, respectively, where $\mathbf{d}_1\in\mathbb{R}^{R_1}$ and  $\mathbf{d}_2\in\mathbb{R}^{R_2}$. Motivated by the rate-splitting technique used in the celebrated Han-Kobayashi schemes, we divide the $R_1$ symbols in $\mathbf{d}_1$ into two parts: the common information $\mathbf{d}_{1c}\in\mathbb{R}^{R_{1c}}$, which is decodable at both $t_1$ and $t_2$,  and the private information $\mathbf{d}_{1p}\in\mathbb{R}^{R_{1p}}$, which is decodable at $t_1$ only. Then we have $\mathbf{d}_1=\left[\begin{matrix}\mathbf{d}_{1c}\\ \mathbf{d}_{1p}\end{matrix}\right]$ and $R_1=R_{1c}+R_{1p}$. To map $\mathbf d_1$ to the $m_1$-dimensional input vector $\mathbf x_1'$,  random spreading is applied. Specifically, let $E_{1c}\in\mathbb{R}^{r_{21}\times{R_{1c}}}$, $E_{1p}\in\mathbb{R}^{(m_1-r_{21})\times R_{1p}}$ be randomly and independently generated matrices, then the transmitted signal vector in the channel model \eqref{eq:mch1} is given by $\mathbf{x}_1'=\left[\begin{matrix}E_{1c}\mathbf{d}_{1c} \\ E_{1p}\mathbf{d}_{1p}\end{matrix}\right]$. Note that $\mathbf{x}_1'$ is obtained by effectively precoding the information-bearing symbols $\mathbf d_1$ with a \emph{block-diagonal} matrix with $E_{1c}$ and $E_{1p}$ on the block diagonals. The effect is that the common and private symbols $\mathbf d_{1c}$ and $\mathbf d_{1p}$ are constrained to the first $r_{21}$ and last $m_1-r_{21}$ components of $\mathbf{x}_1'$, respectively. With such a rate-splitting and precoding, later we will show that the private symbols $\mathbf d_{1p}$ will not affect the received signal vector $\mathbf y_2'$ at $t_2$, i.e., the inter-user interference caused by the private symbols is zero-forced; and the common symbols $\mathbf d_{1c}$ can be decoded at both $t_1$ and $t_2$ if the constraints given in \eqref{eq:m1}-\eqref{eq:m7} are satisfied. 

Similar transmission scheme can be applied at $s_2$, i.e., the information-bearing symbols $\mathbf{d}_2\in \mathbb{R}^{R_2}$ are  split into $\mathbf{d}_{2c}\in \mathbb{R}^{R_{2c}}$ and $\mathbf{d}_{2p}\in \mathbb{R}^{R_{2p}}$, and $\mathbf{x}_2'=\left[\begin{matrix}E_{2c}\mathbf{d}_{2c} \\ E_{2p}\mathbf{d}_{2p}\end{matrix}\right]$, where  $E_{2c} \in \mathbb{R}^{r_{12}\times R_{2c}}$ and $E_{2p} \in \mathbb{R}^{(m_2-r_{12})\times R_{2p}}$ are randomly and independently generated matrices.

The  channel $H_{11}'$ in \eqref{eq:mch1} can be expressed as
\begin{equation*}
\begin{aligned}
H_{11}'
&=U_1^{T}H_{11}V_2=\left[\begin{matrix}U_{11}^T \\ U_{10}^{T}\end{matrix}\right]H_{11}[\begin{matrix}V_{21} & V_{20}\end{matrix}] \\
&=\left[\begin{matrix}U_{11}^TH_{11}V_{21} & U_{11}^TH_{11}V_{20} \\ U_{10}^TH_{11}V_{21} & U_{10}^TH_{11}V_{20}\end{matrix}\right]
\end{aligned}
\end{equation*}
Therefore, the output at $t_1$ given in \eqref{eq:mch1} can be written as
\begin{align}
\mathbf{y}_1'
=&\left[\begin{matrix}U_{11}^TH_{11}V_{21} & U_{11}^TH_{11}V_{20} \\ U_{10}^TH_{11}V_{21} & U_{10}^TH_{11}V_{20}\end{matrix}\right]
\left[\begin{matrix}E_{1c}\mathbf{d}_{1c} \\ E_{1p}\mathbf{d}_{1p}\end{matrix}\right] \nonumber \\
&+\left[\begin{matrix}D_{12} & 0_{r_{12}\times(m_2-r_{12})}\\ 0_{(n_1-r_{12})\times r_{12}} & 0_{(n_1-r_{12})\times(m_2-r_{12})}\end{matrix}\right]
\left[\begin{matrix}E_{2c}\mathbf{d}_{2c} \\ E_{2p}\mathbf{d}_{2p}\end{matrix}\right]\nonumber\\
=&\left[\begin{matrix}U_{11}^TH_{11}V_{21} & U_{11}^TH_{11}V_{20} & D_{12}\\ U_{10}^TH_{11}V_{21} & U_{10}^TH_{11}V_{20} & 0\end{matrix}\right]
\left[\begin{matrix}E_{1c}\mathbf{d}_{1c}\\E_{1p}\mathbf{d}_{1p}\\E_{2c}\mathbf{d}_{2c}\end{matrix}\right] \label{eq:y1Prime}
\end{align}
\eqref{eq:y1Prime} clearly shows that the private symbol vector $\mathbf d_{2p}$ transmitted by $s_2$ does not affect $\mathbf{y}_1'$ due to the block-diagonal precoding discussed previously. Although $t_1$ is interested in recovering $\mathbf{d}_{1c}$ and $\mathbf{d}_{1p}$ only, \eqref{eq:y1Prime} shows that decoding the common symbols $\mathbf d_{2c}$ is also necessary since the decoding process is equivalent to solving a system of linear equations with unknowns $\mathbf{d}_{1c}$, $\mathbf{d}_{1p}$ and $\mathbf d_{2c}$.


Similar arguments hold for $t_2$ as well, where we have 
\begin{align}
\mathbf{y}_2'
=&\left[\begin{matrix}U_{21}^TH_{22}V_{11} & U_{21}^TH_{22}V_{10} \\ U_{20}^TH_{22}V_{11} & U_{20}^TH_{22}V_{10}\end{matrix}\right]
\left[\begin{matrix}E_{2c}\mathbf{d}_{2c} \\ E_{2p}\mathbf{d}_{2p}\end{matrix}\right]\nonumber \\
&+\left[\begin{matrix}D_{21} & 0_{r_{21}\times(m_1-r_{21})}\\ 0_{(n_2-r_{21})\times r_{21}} & 0_{(n_2-r_{21})\times(m_1-r_{21})}\end{matrix}\right]
\left[\begin{matrix}E_{1c}\mathbf{d}_{1c} \\ E_{1p}\mathbf{d}_{1p}\end{matrix}\right]\nonumber\\
=&\left[\begin{matrix}U_{21}^TH_{22}V_{11} & U_{21}^TH_{22}V_{10} & D_{21}\\ U_{20}^TH_{22}V_{11} & U_{20}^TH_{22}V_{10} & 0\end{matrix}\right]\left[\begin{matrix}E_{2c}\mathbf{d}_{2c}\\E_{2p}\mathbf{d}_{2p}\\E_{1c}\mathbf{d}_{1c}\end{matrix}\right] \label{eq:y2Prime}
\end{align}


To find a sufficient condition such that the system of linear equations given by \eqref{eq:y1Prime} and \eqref{eq:y2Prime} are uniquely solvable, the following results are shown to be useful:

\begin{lemma}\cite{Bernstein08}\label{lem:uniqueDec}
Given the relationship $\mathbf{y}=A\mathbf{x}$, where $A\in\mathbb{R}^{p\times l}$, $\mathbf{x}\in\mathbb{R}^{l}$ and $\mathbf y \in \mathbb{R}^{p}$, 
then $\mathbf{x}$ can be uniquely determined from $\mathbf{y}$ if $A$ is of full column rank, i.e., $\mathrm{rank}(A)=l$.
\end{lemma}

To simplify the presentation, let
\begin{align}
M_1=\left[\begin{matrix}U_{11}^TH_{11}V_{21} & U_{11}^TH_{11}V_{20} & D_{12}\\ U_{10}^TH_{11}V_{21} & U_{10}^TH_{11}V_{20} & 0\end{matrix}\right]
\left[\begin{matrix}E_{1c}&0&0\\0&E_{1p}&0\\0&0&E_{2c}\end{matrix}\right]\label{eq:TransM1}
\end{align}

Hence, \eqref{eq:y1Prime} can be written as $\mathbf{y}_1'=M_1\left[\begin{matrix}\mathbf{d}_{1c}^{T}& \mathbf{d}_{1p}^{T} & \mathbf{d}_{2c}^{T}\end{matrix}\right]^T$. According to Lemma \ref{lem:uniqueDec}, the receiver $t_1$ can successfully decode $\mathbf{d}_{1c}$, $\mathbf{d}_{1p}$ and $\mathbf{d}_{2c}$ if $M_1$ is of full column rank. Next, we find a sufficient condition over the data rates $R_{1c}, R_{1p}$ and $R_{2c}$ such that $M_1$ has full column rank.

\begin{lemma}\label{lem:MutliplyRank3}
Given  $A_1\in\mathbb{R}^{p\times l_1}$, $A_2\in\mathbb{R}^{p\times l_2}$ and $A_3\in\mathbb{R}^{p\times l_3}$, and let $E_1\in \mathbb{R}^{l_1\times k_1}$, $E_2 \in \mathbb{R}^{l_2\times k_2}$ and $E_3 \in \mathbb{R}^{l_3\times k_3}$ be  randomly and independently generated,  then the full column rank condition $\mathrm{rank}([\begin{matrix}A_1E_1 & A_2E_2 & A_3E_3\end{matrix}])=k_1+k_2+k_3$ holds with probability 1 if the following conditions are satisfied:
\begin{itemize}
\item{$k_1\leq\mathrm{rank}(A_1)$}
\item{$k_2\leq\mathrm{rank}(A_2)$}
\item{$k_3\leq\mathrm{rank}(A_3)$}
\item{$k_1+k_2\leq\mathrm{rank}([\begin{matrix}A_1 & A_2\end{matrix}])$}
\item{$k_1+k_3\leq\mathrm{rank}([\begin{matrix}A_1 & A_3\end{matrix}])$}
\item{$k_2+k_3\leq\mathrm{rank}([\begin{matrix}A_2 & A_3\end{matrix}])$}
\item{$k_1+k_2+k_3\leq\mathrm{rank}([\begin{matrix}A_1 & A_2 & A_3\end{matrix}])$}
\end{itemize}
\end{lemma}
\begin{proof}
Please refer to Appendix~\ref{A:fullrank}.
\end{proof}

Consider $M_1$ given in \eqref{eq:TransM1}, we have
\begin{align*}
M_1=\left[\begin{matrix}\left(\begin{matrix}U_{11}^TH_{11}V_{21}\\U_{10}^TH_{11}V_{21}\end{matrix}\right)E_{1c} & \left(\begin{matrix}U_{11}^TH_{11}V_{20}\\U_{10}^TH_{11}V_{20}\end{matrix}\right)E_{1p} & \left(\begin{matrix}D_{12}\\0\end{matrix}\right)E_{2c}\end{matrix}\right]
\end{align*}

Recall that $E_{1c}\in \mathbb{R}^{r_{21}\times R_{1c}}$, $E_{1p}\in \mathbb{R}^{(m_1-r_{21})\times R_{1p}}$ and $E_{2c}\in \mathbb{R}^{r_{12}\times R_{2c}}$. By directly applying Lemma~\ref{lem:MutliplyRank3}, a sufficient condition for $M_1$ to be of full column rank, and hence \eqref{eq:y1Prime} is uniquely solvable,  is given by

\begin{align}
R_{1c}&\leq \mathrm{rank}\left(\left[\begin{matrix}U_{11}^TH_{11}V_{21}\\ U_{10}^TH_{11}V_{21}\end{matrix}\right]\right)\nonumber\\
&=\mathrm{rank}(U_{1}^TH_{11}V_{21})\label{eq:cp11}\\
R_{1p}&\leq \mathrm{rank}\left(\left[\begin{matrix}U_{11}^TH_{11}V_{20}\\ U_{10}^TH_{11}V_{20}\end{matrix}\right]\right)\nonumber\\
&=\mathrm{rank}(U_{1}^{T}H_{11}V_{20})\overset{(a)}{=}m_1-r_{21}\label{eq:cp12}\\
R_{2c}&\leq \left(\left[\begin{matrix}D_{12}\\ 0_{(n_1-r_{12})\times r_{12}}\end{matrix}\right]\right)=r_{12}\label{eq:cp13}\\
R_{1c}+R_{1p}&\leq \mathrm{rank}\left(\left[\begin{matrix}U_{11}^TH_{11}V_{21} & U_{11}^TH_{11}V_{20}\\ U_{10}^TH_{11}V_{21} & U_{10}^TH_{11}V_{20}\end{matrix}\right]\right)\nonumber\\
&=\mathrm{rank}(U_{1}^{T}H_{11}V_{2})\overset{(b)}{=}r_{11}\label{eq:cp14}\\
R_{1c}+R_{2c}&\leq \mathrm{rank}\left(\left[\begin{matrix}U_{11}^TH_{11}V_{21} & D_{12}\\ U_{10}^TH_{11}V_{21} & 0_{(n_1-r_{12})\times r_{12}}\end{matrix}\right]\right)\nonumber\\
&\overset{(c)}{=}r_{12}+\mathrm{rank}(U_{10}^TH_{11}V_{21})\label{eq:cp15}\\
R_{1p}+R_{2c} &\leq \mathrm{rank}\left(\left[\begin{matrix}U_{11}^TH_{11}V_{20} & D_{12}\\ U_{10}^TH_{11}V_{20} & 0_{(n_1-r_{12})\times r_{12}}\end{matrix}\right]\right)\nonumber\\
&\overset{(d)}{=}r_{12}+\mathrm{rank}(U_{10}^TH_{11}V_{20})\label{eq:cp16}\\
R_{1c}+R_{1p}+R_{2c}&\leq \mathrm{rank}\left(\left[\begin{matrix}U_{11}^TH_{11}V_{21} & U_{11}^TH_{11}V_{20} & D_{12}\\ U_{10}^TH_{11}V_{21} & U_{10}^TH_{11}V_{20} & 0\end{matrix}\right]\right)\nonumber\\
&\overset{(e)}{=}n_1\label{eq:cp17}
\end{align}
where (a) follows from Lemma \ref{lem:2matrixMultiply} in Appendix~\ref{A:usefulLemma} as
\begin{equation*}
\begin{aligned}
&\mathrm{rank}(U_{1}^TH_{11}V_{20})\overset{(f)}{=}\mathrm{rank}(H_{11}V_{20})\\
&=\mathrm{rank}(V_{20})-\mathrm{dim}(\mathcal{N}(H_{11})\cap\mathcal{R}(V_{20}))\\
&=\mathrm{rank}(V_{20})-\mathrm{dim}(\mathcal{N}(H_{11})\cap\mathcal{N}(H_{21}))\\
&\overset{(g)}=\mathrm{rank}(V_{20})=m_1-r_{21}
\end{aligned}
\end{equation*}
where (f) holds since $U_{1}$ is full rank, and (g) follows the assumption given in \eqref{eq:asm3}.

Moreover, (b) is true since $U_{1}$ and $V_{2}$ are non-singular, (c) and (d) can be obtained by applying elementary column operations since $D_{12}$ is nonsingular, (e) can be shown with elementary column operations together with similar proof as that for (a), i.e.
\begin{align*}
&\mathrm{rank}\left(\left[\begin{matrix}U_{11}^TH_{11}V_{21} & U_{11}^TH_{11}V_{20} & D_{12}\\ U_{10}^TH_{11}V_{21} & U_{10}^TH_{11}V_{20} & 0\end{matrix}\right]\right)\\
&=\mathrm{rank}\left(\left[\begin{matrix} U_{10}^TH_{11}V_{21} & U_{10}^TH_{11}V_{20} \end{matrix}\right]\right)+\mathrm{rank}(D_{12})\\
&=\mathrm{rank}(U_{10}^TH_{11}V_2)+r_{12}\\
&=\mathrm{rank}(U_{10}^TH_{11})+r_{12}\\
&=\mathrm{rank}(H_{11})-\mathrm{dim}(\mathcal{N}(U_{10}^T)\cap\mathcal{R}(H_{11}))+r_{12}\\
&=r_{11}-\mathrm{dim}(\mathcal{R}(H_{12})\cap\mathcal{R}(H_{11}))+r_{12}\\
&=r_{11}-(r_{11}+r_{12}-\mathrm{rank}([\begin{matrix}H_{11} & H_{12}\end{matrix}]))+r_{12}\\
&=\mathrm{rank}([\begin{matrix}H_{11} & H_{12}\end{matrix}])=n_1
\end{align*}

By symmetry, a sufficient condition for receiver $t_2$ to successfully decode $\mathbf{d}_{2c}$, $\mathbf{d}_{2p}$ and $\mathbf{d}_{1c}$ is given by
\begin{align}
R_{2c}&\leq \mathrm{rank}(U_{2}^TH_{22}V_{11})\label{eq:cp21}\\
R_{2p}&\leq m_2-r_{12}\label{eq:cp22}\\
R_{1c}&\leq r_{21}\label{eq:cp23}\\
R_{2c}+R_{2p}&\leq r_{22}\label{eq:cp24}\\
R_{2c}+R_{1c}&\leq r_{21}+\mathrm{rank}(U_{20}^TH_{22}V_{11})\label{eq:cp25}\\
R_{2p}+R_{1c} &\leq r_{21}+\mathrm{rank}(U_{20}^TH_{22}V_{10})\label{eq:cp26}\\
R_{2c}+R_{2p}+R_{1c}&\leq n_2\label{eq:cp27}
\end{align}

Since $R_1=R_{1c}+R_{1p}$ and $R_2=R_{2c}+R_{2p}$, the conditions on the data rate $R_1$ and $R_2$ to ensure full decodability at the respective destinations can be obtained using Fourier-Motzkin Elimination over \eqref{eq:cp11}-\eqref{eq:cp17} and \eqref{eq:cp21}-\eqref{eq:cp27}. The detailed steps can be found in \cite{FM12} and the resulted achievable rate region is given by \eqref{eq:m1}-\eqref{eq:m7}. This completes the achievability proof.

To sum up, the precoding and decoding process in the proposed achievable scheme is depicted in Fig.\ref{F:EncDec}.

\begin{figure}[htb]
\centering
\includegraphics[scale=0.7]{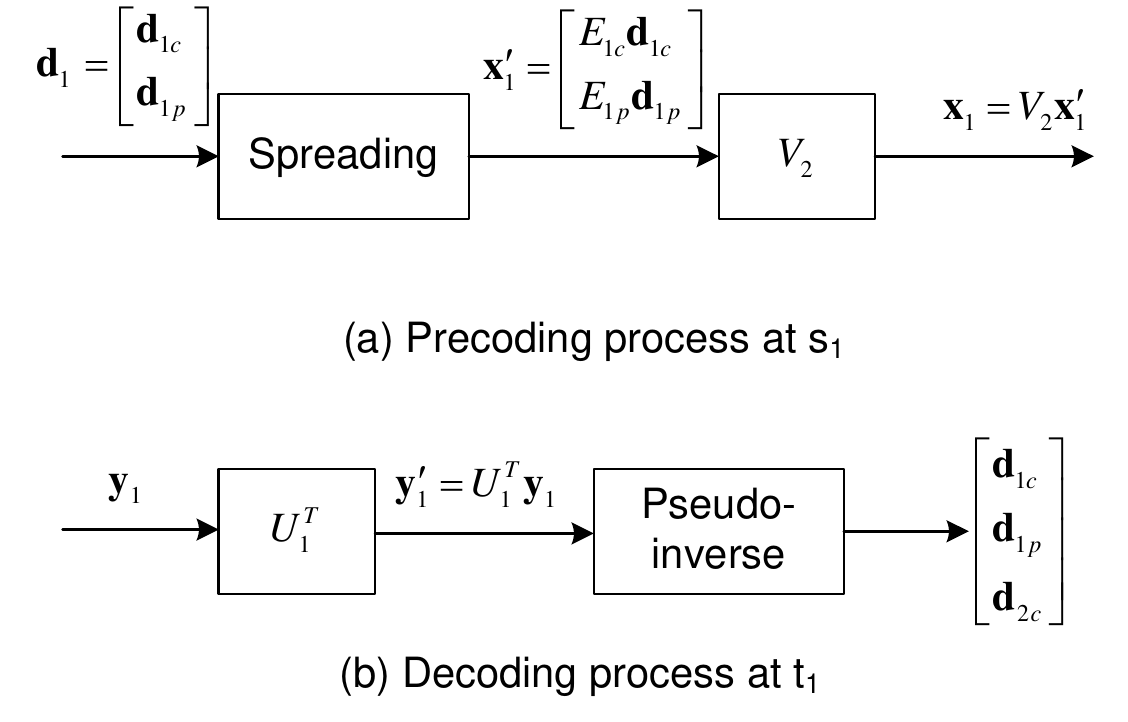}
\caption{Precoding and Decoding at $s_1$ and $t_1$ respectively}
\label{F:EncDec}
\end{figure}

\subsection{Proof of the converse}
In this subsection, we give the converse proof of Theorem~\ref{tho:main}. Firstly, it is obvious that the rates for the two-user linear deterministic IC given by \eqref{eq:DIC} are bounded by the single-user capacity, i.e. $R_1\leq r_{11}, R_2\leq r_{22}$.

For notational convenience, let $\mathbf{z}_1=H_{21}\mathbf{x}_1$, $\mathbf{z}_2=H_{12}\mathbf{x}_2$, which are the interference terms observed at $t_2$ and $t_1$, respectively. 

For $t_1$ to fully decode $\mathbf{x}_1$, it is obtained that the interference term $\mathbf{z}_2=\mathbf{y}_1-H_{11}\mathbf{x}_1$ can be uniquely determined as well. Since $V_{10}$ spans the null space of $H_{12}$, we can obtain  $\mathrm{rank}\left(\left[\begin{matrix}H_{12} \\ V_{10}^T\end{matrix}\right]\right)=m_2$, which is of full column rank. Therefore, if the term $\left[\begin{matrix}H_{12} \\ V_{10}^T\end{matrix}\right]\mathbf{x}_2$ is given, then $\mathbf{x}_2$ can be uniquely determined. As $H_{12}\mathbf x_2=\mathbf z_2$ is known, if a ``genie" provides the value for $V_{10}^T\mathbf{x}_2$ to the receiver $t_1$, then both $\mathbf{x}_1$ and $\mathbf{x}_2$ are decodable. Let $H(\cdot)$ denote the entropy, then the sum rate must satisfy
\begin{align*}
&R_1+R_2\leq H\left(\left[\begin{matrix}\mathbf{y}_1 \\ V_{10}^T\mathbf{x}_2\end{matrix}\right]\right)\\
&=H\left(\left[\begin{matrix}H_{11} & H_{12} \\0_{(m_2-r_{12})\times m_1} & V_{10}^T\end{matrix}\right]\left[\begin{matrix}\mathbf{x}_1 \\ \mathbf{x}_2\end{matrix}\right]\right)\\
&\leq \mathrm{rank}\left(\left[\begin{matrix}H_{11} & H_{12} \\0_{(m_2-r_{12})\times m_1} & V_{10}^T\end{matrix}\right]\right)\\
&\leq \mathrm{rank}\left(\left[\begin{matrix}H_{11} & H_{12}\end{matrix}\right]\right)+\mathrm{rank}\left(\left[\begin{matrix}0_{(m_2-r_{12})\times m_1} & V_{10}^T\end{matrix}\right]\right)\\
&=n_1+m_2-r_{12}
\end{align*}
This completes the proof of \eqref{eq:3}.  \eqref{eq:4} can be proved similarly.

To prove \eqref{eq:5}, the following result is needed:
\begin{lemma}\label{lem:entropyMatrix}
Let $\mathbf{x}$ be a random vector of dimension $l\times 1$, given $A\in\mathbb{R}^{p_1\times l},B\in\mathbb{R}^{p_2\times l}$, then $H(A\mathbf{x}|B\mathbf{x})\leq\mathrm{rank}\left(\left[\begin{matrix}A \\B\end{matrix}\right]\right)-\mathrm{rank}(B)$.
\end{lemma}
\begin{proof}
Please refer to Appendix~\ref{A:entropyMatrix}.
\end{proof}

Since $\mathbf{x}_1$ and $\mathbf{z}_2$ can be simultaneously determined from $\mathbf{y}_1$ at $t_1$, we have $H(\mathbf{x}_1, \mathbf{z}_2)\leq H(\mathbf{y}_1)$. Furthermore, as $\mathbf{x}_1$ and $\mathbf{z}_2$ are independent, we have $H(\mathbf{x}_1, \mathbf{z}_2)=H(\mathbf{x}_1)+H(\mathbf{z}_2)$, which gives  $H(\mathbf{x}_1)\leq H(\mathbf{y}_1)-H(\mathbf{z}_2)$. Similarly, it can be shown that $H(\mathbf{x}_2)\leq H(\mathbf{y}_2)-H(\mathbf{z}_1)$. Therefore, we have
\begin{align*}
R_1&+R_2\leq H(\mathbf{x}_1)+H(\mathbf{x}_2)\\
\leq &H(\mathbf{y}_1)-H(\mathbf{z}_2)+H(\mathbf{y}_2)-H(\mathbf{z}_1)\\
\leq &H(\mathbf{y}_1, \mathbf z_1)-H(\mathbf{z}_1)+H(\mathbf{y}_2, \mathbf z_2)-H(\mathbf{z}_2)\\
= &H(\mathbf{y}_1|\mathbf{z}_1)+H(\mathbf{y}_2|\mathbf{z}_2)\\
\overset{(a)}{\leq}&\mathrm{rank}\left(\left[\begin{matrix}H_{11} & H_{12}\\ H_{21} & 0\end{matrix}\right]\right)
-\mathrm{rank}\left(\left[\begin{matrix}H_{21} & 0\end{matrix}\right]\right)
\\&+\mathrm{rank}\left(\left[\begin{matrix}H_{21} & H_{22} \\ 0 & H_{12}\end{matrix}\right]\right)-\mathrm{rank}\left(\left[\begin{matrix}0 & H_{12}\end{matrix}\right]\right)\\
=&\mathrm{rank}\left(\left[\begin{matrix}H_{11} & H_{12}\\ H_{21} & 0\end{matrix}\right]\right)+\mathrm{rank}\left(\left[\begin{matrix}H_{21} & H_{22} \\ 0 & H_{12}\end{matrix}\right]\right)-r_{12}-r_{21}
\end{align*}
where (a) follows from Lemma \ref{lem:entropyMatrix}. This completes the proof of \eqref{eq:5}.

To prove \eqref{eq:6}, similar arguments can be used, i.e.,
\begin{align*}
2R_1&+R_2\leq 2H(\mathbf{x}_1)+H(\mathbf{x}_2)\\
\leq & 2H(\mathbf{y}_1)-2H(\mathbf{z}_2)+H(\mathbf{y}_2)-H(\mathbf{z}_1)\\
\leq & H(\mathbf{y}_1)+H(\mathbf{y}_1, \mathbf z_1, \mathbf z_2)-2H(\mathbf{z}_2)+H(\mathbf{y}_2)-H(\mathbf{z}_1)\\
= & H(\mathbf{y}_1)+H(\mathbf{y}_1|\mathbf{z}_1,\mathbf{z}_2)+H(\mathbf{y_2}|\mathbf{z}_2)\\
\overset{(b)}{\leq} & n_1+\mathrm{rank}\left(\left[\begin{matrix}H_{11} & H_{12}\\H_{21} & 0\\ 0 & H_{12}\end{matrix}\right]\right)-\mathrm{rank}\left(\left[\begin{matrix}H_{21} & 0\\ 0 & H_{12}\end{matrix}\right]\right)\\
&+\mathrm{rank}\left(\left[\begin{matrix}H_{21} & H_{22} \\ 0 & H_{12}\end{matrix}\right]\right)-r_{12}\\
=&n_1+m_1+\mathrm{rank}\left(\left[\begin{matrix}H_{21} & H_{22} \\ 0 & H_{12}\end{matrix}\right]\right)-r_{12}-r_{21}
\end{align*}
where (b) follows from Lemma \ref{lem:entropyMatrix}. By symmetry, \eqref{eq:7} can be proved similarly.

This completes the proof of the converse for Theorem~\ref{tho:main}.
\subsection{Discussion}

In this subsection, we show that the DoF region for the two-user Gaussian IC given in \eqref{eq:JafarBound} is a special case of the capacity region for the linear deterministic IC specified in \eqref{eq:1}-\eqref{eq:7}, which is obtained with $H_{ij},i,j\in\{1,2\}$ being random (thus full rank) and independent.

If all the channel matrices are full rank, we have $r_{11}=\min\{m_1,n_1\}$, $r_{22}=\min\{m_2,n_2\}$, $r_{12}=\min\{m_2,n_1\}$, and $r_{21}=\min\{m_1,n_2\}$. Since they are also independent,  $\mathrm{rank}\left(\left[\begin{matrix}H_{11} & H_{12} \\ H_{21} & 0\end{matrix}\right]\right)=\mathrm{rank}\left(\left[\begin{matrix}H_{21} & H_{22} \\ 0& H_{12}\end{matrix}\right]\right)=\min\{m_1+m_2,n_1+n_2\}$. Thus, \eqref{eq:1}-\eqref{eq:7} reduce to
\begin{align}
R_1\leq&\min\{m_1,n_1\}\label{eq:nondegenerated1}\\
R_2\leq&\min\{m_2,n_2\}\label{eq:nondegenerated2}\\
R_1+R_2\leq &n_1+m_2-\min\{m_2,n_1\}=\max\{m_2,n_1\}\label{eq:nondegenerated3}\\
R_1+R_2\leq &n_2+m_1-\min\{m_1,n_2\}=\max\{m_1,n_2\}\label{eq:nondegenerated4}\\
R_1+R_2\leq &2\min\{m_1+m_2,n_1+n_2\}\nonumber\\
&-\min\{m_2,n_1\}-\min\{m_1,n_2\}\label{eq:nondegenerated5}\\
2R_1+R_2\leq &m_1+n_1+\min\{m_1+m_2,n_1+n_2\}\nonumber\\
&-\min\{m_2,n_1\}-\min\{m_1,n_2\}\label{eq:nondegenerated6}\\
R_1+2R_2\leq &m_2+n_2+\min\{m_1+m_2,n_1+n_2\}\nonumber\\
&-\min\{m_2,n_1\}-\min\{m_1,n_2\}\label{eq:nondegenerated7}
\end{align}

If $m_1\geq n_2$, $m_2\geq n_1$ and $m_1+m_2\geq n_1+n_2$, \eqref{eq:nondegenerated1}-\eqref{eq:nondegenerated7} reduce to
\begin{align}
R_1&\leq\min\{m_1,n_1\}\label{eq:reduce1}\\
R_2&\leq\min\{m_2,n_2\}\label{eq:reduce2}\\
R_1+R_2&\leq m_2\label{eq:reduce3}\\
R_1+R_2&\leq m_1\label{eq:reduce4}\\
R_1+R_2&\leq n_1+n_2\label{eq:reduce5}\\
2R_1+R_2&\leq m_1+n_1\label{eq:reduce6}\\
R_1+2R_2&\leq m_2+n_2\label{eq:reduce7}
\end{align}
Note that \eqref{eq:reduce6} is implied by \eqref{eq:reduce1} and \eqref{eq:reduce4}, and thus is redundant. Similarly, \eqref{eq:reduce7} is also redundant. The region formed by \eqref{eq:reduce1}-\eqref{eq:reduce7} is exactly the same as the that given in \eqref{eq:JafarBound} under the same conditions that $m_1\geq n_2$, $m_2\geq n_1$ and $m_1+m_2\geq n_1+n_2$. Other cases can be proved in a similar manner and thus are omitted here for brevity.

\section{An Achievable Region for Double-Unicast Networks}\label{sec:nc}
A double-unicast network can be represented by a directed acyclic graph $G=(V,E)$ with two sources $S=\{s_1,s_2\}\subset V$ and two receivers $T=\{t_1,t_2\}\subset V$. Similar to the two-user ICs, $s_1$ and $s_2$ are intended to send independent messages to $t_1$ and $t_2$, respectively, and inter-user interference is resulted since both pairs share the same network.

Assuming that each edge is capable of carrying one symbol per time slot. By Merger's theorem, the minimum cut between sets $S_{N_1}\subseteq S$ and $T_{N_2}\subseteq T$ is the number of edge disjoint paths from $S_{N_1}$ to $T_{N_2}$ (denoted by $k_{N_1-N_2}$), where $N_1,N_2\subseteq\{1,2\}$. With random linear network coding performed at all intermediate nodes, the double-unicast network can be modeled as a two-user linear deterministic IC where the transition matrices are determined by the network topology and the coding coefficients chosen at each intermediate node. Without loss of generality, the dimension and the rank of the transition matrices are represented by the min-cuts of the network, i.e., $k_{N_1-N_2}$ as shown in Table \ref{table1} \cite{Huang11}.\\
\begin{table}
\caption{Dimension and Rank of Matrices}
\centering
\begin{tabular}{|l || c |r|}
  \hline
   Channel Matrix & Size & Rank \\
   \hline
  $H_{11}$ & $k_{12-1}\times k_{1-12}$ & $k_{1-1}$ \\
  \hline
  $H_{12}$ & $k_{12-1}\times k_{2-12}$ & $k_{2-1}$ \\
  \hline
  $H_{21}$ & $k_{12-2}\times k_{1-12}$ & $k_{1-2}$ \\
  \hline
  $H_{22}$ & $k_{12-2}\times k_{2-12}$ & $k_{2-2}$ \\
  \hline
  $\left[\begin{matrix}H_{11} & H_{12}\end{matrix}\right]$ & $k_{12-1}\times(k_{1-12}+k_{2-12})$ & $k_{12-1}$ \\
  \hline
  $\left[\begin{matrix}H_{21} & H_{22}\end{matrix}\right]$ & $k_{12-2}\times(k_{1-12}+k_{2-12})$ & $k_{12-2}$ \\
  \hline
  $\left[\begin{matrix}H_{11} \\ H_{21}\end{matrix}\right]$ & $(k_{12-1}+k_{12-2})\times k_{1-12}$ & $k_{1-12}$\\
  \hline
  $\left[\begin{matrix}H_{12} \\ H_{22}\end{matrix}\right]$ & $(k_{12-1}+k_{12-2})\times k_{2-12}$ & $k_{2-12}$\\
  \hline
\end{tabular}
\label{table1}
\end{table}

With Theorem \ref{tho:main}, the rate pair $(R_1,R_2)$ is achievable if  the following conditions are satisfied:
\begin{align}
R_1\leq &k_{1-1}\label{eq:ncbound1}\\
R_2\leq &k_{2-2}\label{eq:ncbound2}\\
R_1+R_2\leq &k_{12-1}+k_{2-12}-k_{2-1}\label{eq:ncbound4}\\
R_1+R_2\leq &k_{12-2}+k_{1-12}-k_{1-2}\label{eq:ncbound3}\\
R_1+R_2\leq &\mathrm{rank}\left(\left[\begin{matrix}H_{11} & H_{12} \\
H_{21} & 0\end{matrix}\right]\right)+\mathrm{rank}\left(\left[\begin{matrix}H_{21} & H_{22} \\
0 & H_{12}\end{matrix}\right]\right)\nonumber\\
&-k_{1-2}-k_{2-1}\label{eq:ncbound5}\\
2R_1+R_2\leq &k_{12-1}+k_{1-12}+\mathrm{rank}\left(\left[\begin{matrix}H_{21} & H_{22} \\
0 & H_{12}\end{matrix}\right]\right)\nonumber\\
&-k_{1-2}-k_{2-1}\label{eq:ncbound6}\\
R_1+2R_2\leq &k_{12-2}+k_{2-12}+\mathrm{rank}\left(\left[\begin{matrix}H_{11} & H_{12} \\
H_{21} & 0\end{matrix}\right]\right)\nonumber\\
&-k_{1-2}-k_{2-1}\label{eq:ncbound7}
\end{align}

Note that although \eqref{eq:1}-\eqref{eq:7} give the capacity region of the linear deterministic IC, \eqref{eq:ncbound1}-\eqref{eq:ncbound7} only give an achievable region for double-unicast networks since random linear network coding at all the intermediate nodes  may be sub-optimal. 
Nevertheless, as random network coding can be practically implemented due to its simplicity, it is widely used in practice \cite{Erez09,Huang11,Das10}.

\subsection{Comparison with Existing Results}
In this subsection, the achievable region given by \eqref{eq:ncbound1}-\eqref{eq:ncbound7} is compared with that in \cite{Huang11} and \cite{Erez09}, which both consider the double-unicast networks.

\subsubsection{Comparison with \cite{Huang11}}
In \cite{Huang11},  two scenarios are considered separately, i.e., the low interference case with $k_{1-2}+k_{2-1}\leq \min (k_{12-1}, k_{12-2})$ and the high interference case with $k_{1-2}+k_{2-1}\geq \min (k_{12-1}, k_{12-2})$. They give rise to the following achievable regions:

\begin{tabular}{| l | l |}
\hline
  \textnormal{Region 1}&\textnormal{Region 2}  \\
    $R_1 \leq k_{12-1}-k_{2-1}$ & $R_1 \leq k_{1-1}$ \\
    $R_2 \leq k_{12-2}-k_{1-2}$ & $R_2\leq\min\{k_{12-1},k_{12-2}\}-k_{1-2}$ \\
    & $R_1+R_2\leq\mathrm{rank}(\left[\begin{matrix}H_{11} & H_{12}M_2\end{matrix}\right])$\\
    \hline
    \multicolumn{2}{| l |}{Region 3}\\
    \multicolumn{2}{| l |}{$R_1\leq\min\{k_{12-1},k_{12-2}\}-k_{2-1}$}\\
    \multicolumn{2}{| l |}{$R_2\leq k_{2-2}$}\\
    \multicolumn{2}{| l |}{$R_1+R_2\leq \mathrm{rank}(\left[\begin{matrix}H_{21}M_1 & H_{22}\end{matrix}\right])$}\\
    \hline
\end{tabular}


where $M_1$ and $M_2$ are some mapping matrices of size $k_{1-12}\times R_1$ and $k_{2-12}\times R_2$, respectively. It is shown in \cite{Huang11} that for the low interference case, the rate pairs that are in the convex hull of Regions 1, 2 and 3  are achievable, and for the high interference case,  the convex hull of Regions 2 and 3 is achievable. 

\begin{lemma}\label{lem:compare1}
The proposed achievable region  specified in \eqref{eq:ncbound1}-\eqref{eq:ncbound7} for the double-unicast networks is larger than the convex hull of Region 1, 2 and 3.
\end{lemma}
\begin{proof}
The  region specified in \eqref{eq:ncbound1}-\eqref{eq:ncbound7} is convex. Therefore, to show that it is larger than the convex hull of Regions 1, 2 and 3, it is sufficient to show that all the individual Regions 1, 2 and 3 are within the proposed region.

To prove that Region 1 is within the proposed region, we show that any rate pair $(R_1,R_2)$ in Region 1  satisfy \eqref{eq:ncbound1}-\eqref{eq:ncbound7}.
\begin{align}
R_1\leq k_{12-1}-k_{2-1}\overset{(a)}{\leq} k_{1-1} \label{eq:Region1R1}
\end{align}
where (a) follows $k_{12-1}=\mathrm{rank}\left(\left[\begin{matrix}H_{11}&H_{12}\end{matrix}\right]\right)\leq\mathrm{rank}(H_{11})+\mathrm{rank}(H_{12})=k_{1-1}+k_{2-1}$.
Similarly, we have $R_2\leq k_{12-2}-k_{1-2}\leq k_{2-2}$. Furthermore, the following inequalities are satisfied:
\begin{align}
R_1+R_2&\leq k_{12-1}-k_{2-1}+k_{12-2}-k_{1-2}\label{eq:Region1sum}\\
&\leq k_{1-1}+k_{12-2}-k_{1-2}\\
&\leq k_{1-12}+k_{12-2}-k_{1-2}
\end{align}
Thus, \eqref{eq:ncbound3} is satisfied. Similarly,  $R_1+R_2\leq k_{12-1}+k_{2-12}-k_{2-1}$ can be proved. With the inequality given in \eqref{eq:Region1sum} and the following relations,
\begin{align}
\mathrm{rank}\left(\left[\begin{matrix}H_{11} & H_{12} \\
H_{21} & 0\end{matrix}\right]\right)\geq\mathrm{rank}\left(\left[\begin{matrix}H_{11} & H_{12}\end{matrix}\right]\right)=k_{12-1}\label{eq:simplefact1}\\
\mathrm{rank}\left(\left[\begin{matrix}H_{21} & H_{22} \\
0 & H_{12}\end{matrix}\right]\right)\geq\mathrm{rank}\left(\left[\begin{matrix}H_{21} & H_{22}\end{matrix}\right]\right)=k_{12-2}\label{eq:simplefact2}
\end{align}
 \eqref{eq:ncbound5} can be proved. Furthermore, \eqref{eq:ncbound6} is satisfied since
\begin{align*}
2R_1+R_2\overset{(b)}{\leq}&k_{1-1}+k_{12-1}-k_{2-1}+k_{12-2}-k_{1-2}\\
\leq & {k_{1-12}+k_{12-1}-k_{2-1}+k_{12-2}-k_{1-2}}\\
\overset{(c)}{\leq}&\mathrm{rank}\left(\left[\begin{matrix}H_{21} & H_{22} \\
0 & H_{12}\end{matrix}\right]\right)+k_{1-12}+k_{12-1}\\
&-k_{2-1}-k_{1-2}
\end{align*}
where (b) follows from \eqref{eq:Region1R1} and \eqref{eq:Region1sum}, and (c) is true due to \eqref{eq:simplefact2}. By symmetry, it can be shown that  \eqref{eq:ncbound7} is also satisfied. This completes the proof that Region 1 is within our proposed region.

In order to show that Region $2$ is also within our proposed region, the following region (denoted as Region $2'$) is defined:
\begin{align}
R_1&\leq k_{1-1}\\
R_2&\leq \min\{k_{12-1},k_{12-2}\}-k_{1-2}\\
R_1+R_2&\leq k_{12-1}
\end{align}
It is obvious that Region $2'$ is no smaller than Region 2 since $\mathrm{rank}([\begin{matrix}H_{11} & H_{12}M_{2}\end{matrix}])\leq\mathrm{rank}\left(\left[\begin{matrix}H_{11} & H_{12}\end{matrix}\right]\right)=k_{12-1}, \forall M_2$. Then it is sufficient to prove that any rate pair $(R_1,R_2)$ in Region $2'$ satisfy \eqref{eq:ncbound1}-\eqref{eq:ncbound7}, which is given as follows:
\begin{align*}
R_1\leq &k_{1-1}\\
R_2\leq &\min\{k_{12-1},k_{12-2}\}-k_{1-2}\nonumber\\
\leq &k_{12-2}-k_{1-2}\overset{(a)}{\leq} k_{2-2}\\
R_1+R_2\leq &k_{1-1}+\min\{k_{12-1},k_{12-2}\}-k_{1-2}\nonumber\\
\overset{(b)}{\leq} &k_{1-12}+k_{12-2}-k_{1-2}\\
R_1+R_2\leq &k_{12-1}\overset{(c)}{\leq} k_{12-1}+(k_{2-12}-k_{2-1})\\
R_1+R_2\leq &k_{12-1}\overset{(d)}{\leq} \mathrm{rank}\left(\left[\begin{matrix}H_{11} & H_{12} \\
H_{21} & 0_{k_{12-2}\times k_{2-12}}\end{matrix}\right]\right)\\
\overset{(e)}{\leq}&\mathrm{rank}\left(\left[\begin{matrix}H_{11} & H_{12} \\
H_{21} & 0\end{matrix}\right]\right)+\mathrm{rank}\left(\left[\begin{matrix}H_{21} & H_{22} \\
0 & H_{12}\end{matrix}\right]\right)\nonumber\\
&-k_{2-1}-k_{1-2}\\
2R_1+R_2\leq &k_{12-1}+k_{1-1}\leq k_{12-1}+k_{1-12}\\
\overset{(f)}{\leq}&k_{12-1}+k_{1-12}+\mathrm{rank}\left(\left[\begin{matrix}H_{21} & H_{22} \\
0 & H_{12}\end{matrix}\right]\right)\nonumber\\
&-k_{2-1}-k_{1-2}\\
R_1+2R_2\leq & k_{12-1}+\min\{k_{12-1},k_{12-2}\}-k_{1-2}\\
\leq &k_{12-1}+k_{12-2}-k_{1-2}\\
\overset{(g)}{\leq}&\mathrm{rank}\left(\left[\begin{matrix}H_{11} & H_{12} \\
H_{21} & 0 \end{matrix}\right]\right)+k_{12-2}-k_{1-2}\\
\overset{(h)}{\leq}&\mathrm{rank}\left(\left[\begin{matrix}H_{11} & H_{12} \\
H_{21} & 0\end{matrix}\right]\right)+k_{12-2}-k_{1-2}\nonumber\\
&+k_{2-12}-k_{2-1}
\end{align*}
where (a) follows from similar arguments as \eqref{eq:Region1R1}; (b) is true since $ k_{1-1}\leq k_{1-12}$; (c) is satisfied due to $k_{2-1}\leq k_{2-12}$; (d) is obtained by using \eqref{eq:simplefact1}, and (e) follows from Lemma \ref{lem:complexityReduce2}, which gives
\begin{align*}
\mathrm{rank}\left[\begin{matrix}H_{21} & H_{22} \\ 0 & H_{12}\end{matrix}\right]-k_{2-1}-k_{1-2}=\mathrm{rank}(U_{20}^{T}H_{22}V_{10})\geq 0
\end{align*}
(f) can be shown with Lemma \ref{lem:complexityReduce2} in a similar manner; (g) is true due to \eqref{eq:simplefact1}; and (h) is satisfied since $k_{2-12}\geq k_{2-1}$. This completes the proof that Region $2'$ is within the proposed achievable region. Therefore, Region 2 must be within our proposed region too. By symmetry,  Region 3 can be shown to be within the proposed region as well. The details are omitted for brevity.

This completes the proof of Lemma~\ref{lem:compare1}.
\end{proof}

\subsubsection{Comparison with \cite{Erez09}}
The achievable region derived in \cite{Erez09} is given by the convex hull of the following two regions:
\[
\begin{array}{cc}
  \textnormal{Region 4}&\textnormal{Region 5} \\
  \begin{matrix}R_1\leq k_{1-1}\\ 2R_1+R_2\leq k_{2-2}\end{matrix}  & \begin{matrix}R_1+2R_2\leq k_{1-1}\\R_2\leq k_{2-2}\end{matrix}
\end{array}
\]
\begin{lemma}\label{lem:compare2}
The proposed achievable region  specified in \eqref{eq:ncbound1}-\eqref{eq:ncbound7} for the double-unicast networks is 
 larger than the convex hull of Region 4 and Region 5.
\end{lemma}
\begin{proof}
Since the proposed region is convex, it is sufficient to show that both Region 4 and Region 5 are within the proposed region. For any rate pair $(R_1, R_2)$ in Region 4,  \eqref{eq:ncbound1}-\eqref{eq:ncbound7} are satisfied since
\begin{align*}
R_1
\leq &k_{1-1}\\
R_2
\leq &k_{2-2}-2R_1\leq k_{2-2}\\
R_1+R_2
\leq &k_{2-2}-R_1\leq k_{2-2}\leq k_{12-2}\\
\leq &k_{12-2}+k_{1-12}-k_{1-2}\\
R_1+R_2
\leq &k_{2-2}-R_1\leq k_{2-2}\leq k_{2-12}\\
\leq &k_{2-12}+k_{12-1}-k_{2-1}\\
R_1+R_2
\leq &k_{12-2}\leq \mathrm{rank}\left(\left[\begin{matrix}H_{21} & H_{22} \\0_{k_{12-1}\times k_{1-12}} & H_{12}\end{matrix}\right]\right)\\
\overset{(a)}{\leq} &\mathrm{rank}\left(\left[\begin{matrix}H_{21} & H_{22}\\
0 & H_{12}\end{matrix}\right]\right)+\mathrm{rank}\left(\left[\begin{matrix}H_{11} & H_{12} \\
H_{21} & 0\end{matrix}\right]\right)\\
&-k_{2-1}-k_{1-2}\\
2R_1+R_2
\leq &k_{2-2}\leq k_{12-2}\leq \mathrm{rank}\left(\left[\begin{matrix}H_{21} & H_{22} \\
0 & H_{12}\end{matrix}\right]\right)\\
\leq &\mathrm{rank}\left(\left[\begin{matrix}H_{21} & H_{22} \\
0 & H_{12}\end{matrix}\right]\right)+(k_{1-12}-k_{1-2})\\
&+(k_{12-1}-k_{2-1})\\
R_1+2R_2\leq &2k_{2-2}\leq k_{2-12}+k_{12-2}\\
\overset{(b)}{\leq}  &k_{2-12}+k_{12-2}+\mathrm{rank}\left(\left[\begin{matrix}H_{11} & H_{12} \\
H_{21} & 0\end{matrix}\right]\right)\\
&-k_{1-2}-k_{2-1}
\end{align*}
where (a) and (b) follows from Lemma \ref{lem:complexityReduce1}.

By symmetry,  any rate pairs $(R_1,R_2)$ in Region $5$ can be shown to satisfy \eqref{eq:ncbound1}-\eqref{eq:ncbound7}. This completes the proof of Lemma \ref{lem:compare2}.
\end{proof}
\subsubsection{Discussion}
It was pointed out in \cite{Huang11} that for certain network topologies,  there exist some rate pairs that are achievable by the scheme in \cite{Huang11} but not achievable by the scheme in \cite{Erez09}, and vice versa.  According to Lemma \ref{lem:compare1} and Lemma \ref{lem:compare2}, our region is larger than both regions given in \cite{Huang11} and \cite{Erez09}. Therefore, it can be concluded that our proposed region is strictly larger than both of them. Actually, there exists some network instance (for example, the network shown in Fig. \ref{F:2Unicast}) where our proposed region is strictly better than the union of the region given in \cite{Huang11} and  \cite{Erez09}.

\subsection{Implementing the Network Code in Finite Field}
To achieve the region specified in \eqref{eq:ncbound1}-\eqref{eq:ncbound7}, the standard MIMO SVD technique has been used. SVD is well defined for real matrices but not for matrices in finite filed. Therefore, to achieve the proposed region for the double-unicast networks within certain finite filed, some slight modifications for the previously discussed achievability scheme is required. To this end, we need to find matrices that have similar properties in the chosen finite field as the orthogonal matrices obtained via SVD.
Specifically, assuming that all the transition matrices are chosen from a finite field $F_q$, where $q$ is power of prime, the following matrices are defined:
\begin{itemize}
\item{Let $\bar{U}_{11}\in F_q^{k_{12-1}\times k_{2-1}}$ be a basis for $\mathcal{R}(H_{12})$.}
\item{Let $\bar{U}_{10}\in F_q^{k_{12-1}\times(k_{12-1}-k_{2-1})}$ be a basis for $\mathcal{N}(H_{12}^T)$.}
\item{Let $\bar{V}_{11}\in F_q^{k_{2-12}\times k_{2-1}}$ be a basis for $\mathcal{R}(H_{12}^T)$.}
\item{Let $\bar{V}_{10}\in F_q^{k_{2-12}\times(k_{2-12}-k_{2-1})}$ be a basis for $\mathcal {N}(H_{12})$.}
\end{itemize}
Then, we have
\begin{align*}
\left[\begin{matrix}\bar{U}_{11}^T \\ \bar{U}_{10}^T\end{matrix}\right]H_{12}\left[\begin{matrix}\bar{V}_{11} & \bar{V}_{10}\end{matrix}\right]
=\left[\begin{matrix}\bar{U}_{11}^TH_{12}\bar{V}_{11} & \bar{U}_{11}^TH_{12}\bar{V}_{10} \\ \bar{U}_{10}^TH_{12}\bar{V}_{11} & \bar{U}_{10}^TH_{12}\bar{V}_{10}\end{matrix}\right]\\
=\left[\begin{matrix}\bar{U}_{11}^TH_{12}\bar{V}_{11} & 0_{k_{2-1}\times(k_{2-12}-k_{2-1})} \\ 0_{(k_{12-1}-k_{2-1})\times k_{2-1}} & 0_{(k_{12-1}-k_{2-1})\times(k_{2-12}-k_{2-1})}\end{matrix}\right]
\end{align*}
where $\bar{U}_{11}^TH_{12}\bar{V}_{11}\in F_q^{k_{2-1}\times k_{2-1}}$ is of full rank, although it may no longer be diagonal. For notational convenience, denote $\bar{U}_{11}^TH_{12}\bar{V}_{11}$ by $\bar{D}_{12}$.

Similarly, we can obtain $\bar{U}_{21},\bar{U}_{20}, \bar{V}_{21},\bar{V}_{20}$ and $\bar{D}_{21}=\bar{U}_{21}^TH_{21}\bar{V}_{21}$ in finite field. When the field size is sufficiently large, Lemma \ref{lem:MutliplyRank3} holds with high probability (refer to Corollary \ref{cor:FF}). Therefore, following similar arguments as that in section \ref{sec:main},  the region specified in \eqref{eq:ncbound1}-\eqref{eq:ncbound7} can be achieved in finite field as well.
\begin{example}
Consider the network shown in Fig.\ref{F:2Unicast}(a) where each edge has unit capacity.  Assume the field size\footnote{A commonly used field size for random network coding is $2^8$. Here, a small field size is chosen for illustration purposes.} is given by $q=7$. With random network coding performed at intermediate nodes, one possible realization of the effective channel matrices for the equivalent two-user linear deterministic IC are given by
\begin{align*}
&H_{11}=\left[\begin{matrix}2 & 0\\ 2 & 3\end{matrix}\right];H_{12}=\left[\begin{matrix}2&1&0\\2&1&1\end{matrix}\right]\\
&H_{21}=\left[\begin{matrix}1&0\\2&3\\2&3\end{matrix}\right];H_{22}=\left[\begin{matrix}1&0&0\\2&1&0\\2&1&1\end{matrix}\right]
\end{align*}

\begin{figure}[htb]
\centering
\includegraphics[scale=0.7]{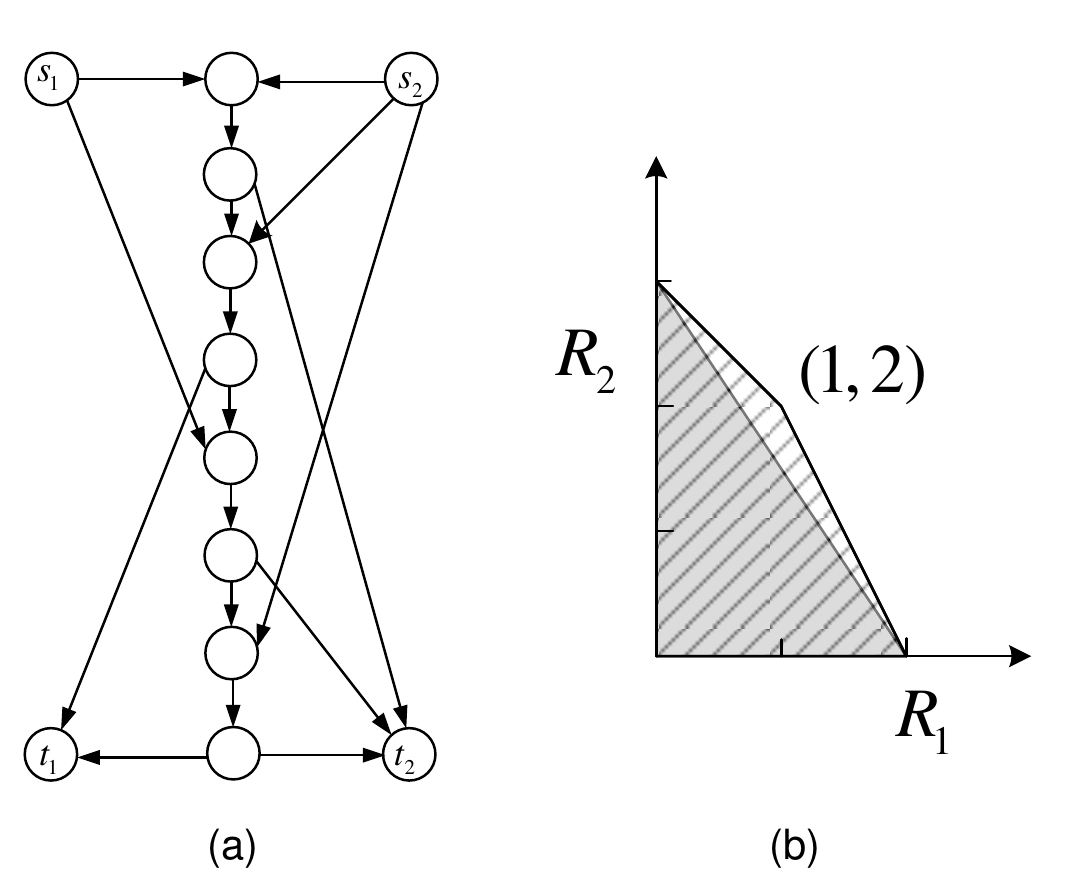}
\caption{(a) An example of double-unicast network (b) Achievable region}
\label{F:2Unicast}
\end{figure}

Then the achievable region (after removing redundant inequalities) specified in \eqref{eq:ncbound1}-\eqref{eq:ncbound7} for this example is given by
\begin{align*}
R_1&\leq 2\\
R_1+R_2&\leq 3\\
2R_1+R_2&\leq 4
\end{align*}
which is plotted in Fig.\ref{F:2Unicast}(b). Note that the gray area denotes the achievable region given by \cite{Huang11} and  \cite{Erez09} with time-sharing. Next, we give the specific precoding scheme to achieve the rate pair $(R_1,R_2)=(1,2)$ by following the achievability scheme given in Section~\ref{sec:main}. Let $\mathbf{d}_1=d_{11}$, $\mathbf{d}_2=\left[\begin{matrix}d_{21} & d_{22}\end{matrix}\right]^T$. Then, following \eqref{eq:cp11}-\eqref{eq:cp16} and \eqref{eq:cp21}-\eqref{eq:cp27}, the private-common rate splitting is given by $\mathbf{d}_{1c}=d_{11}, \mathbf{d}_{1p}=\emptyset, \mathbf{d}_{2c}=d_{21}$ and $\mathbf{d}_{2p}=d_{22}$, where $\emptyset$ denotes an empty vector. 

Assume that the following matrices are randomly generated, which will be applied to map the data symbols to the transmitted vectors:
\begin{align*}
\bar{E}_{1c}=\left[\begin{matrix}4  \\ 3 \end{matrix}\right];
\bar{E}_{2c}=\left[\begin{matrix}2\\3\end{matrix}\right];
\bar{E}_{2p}=\left[\begin{matrix}3\end{matrix}\right]
\end{align*}

$\bar{U}_{11},\bar{U}_{10},\bar{V}_{11},\bar{V}_{10}$ can be determined as well:
\begin{align*}
\bar{U}_{11}=\left[\begin{matrix}1&0\\1&1\end{matrix}\right];
\bar{U}_{10}=0;
\bar{V}_{11}=\left[\begin{matrix}2&2\\1&1\\0&1\end{matrix}\right];
\bar{V}_{10}=\left[\begin{matrix}1\\5\\0\end{matrix}\right]
\end{align*}
Similarly, we can find $\bar U_{21}, \bar U_{20}, \bar V_{21}, \bar V_{20}$
\begin{align*}
\bar{U}_{21}=\left[\begin{matrix}1&0\\2&3\\2&3\end{matrix}\right];
\bar{U}_{20}=\left[\begin{matrix}0\\3\\4\end{matrix}\right];
\bar{V}_{21}=\left[\begin{matrix}1&2\\0&3\end{matrix}\right];
\bar{V}_{22}=0
\end{align*}

Thus, the data transmitted at $s_1$ is given by
\begin{align*}
\mathbf{x}_1=\bar{V}_{21}E_{1c}\mathbf{d}_{1c}
=\left[\begin{matrix}1&2\\0&3\end{matrix}\right]\left[\begin{matrix}4 \\ 3 \end{matrix}\right]d_{11}
=\left[\begin{matrix}3d_{11} \\ 2d_{11}\end{matrix}\right]
\end{align*}

Similarly,
\begin{align*}
\mathbf{x}_2&=\left[\begin{matrix}\bar{V}_{11} & \bar{V}_{10}\end{matrix}\right]\left[\begin{matrix}E_{2c}\mathbf{d}_{2c}\\E_{2p}\mathbf{d}_{2p}\end{matrix}\right]
=\left[\begin{matrix}2&2&1\\1&1&5\\0&1&0\end{matrix}\right]\left[\begin{matrix}2d_{21}\\3d_{21}\\3d_{22}\end{matrix}\right]\\
&=\left[\begin{matrix}3d_{21}+3d_{22}\\5d_{21}+d_{22}\\3d_{21}\end{matrix}\right]
\end{align*}

Therefore, the data received at $t_1$ and $t_2$ are given by
\begin{align*}
\mathbf{y}_1&=H_{11}\left[\begin{matrix}3d_{11} \\ 2d_{11}\end{matrix}\right]+H_{12}\left[\begin{matrix}3d_{21}+3d_{22}\\5d_{21}+d_{22}\\3d_{21}\end{matrix}\right]
=\left[\begin{matrix}6d_{11}+4d_{21}\\5d_{11}\end{matrix}\right]\\
\mathbf{y}_2&=H_{22}\left[\begin{matrix}3d_{21}+3d_{22}\\5d_{21}+d_{22}\\3d_{21}\end{matrix}\right]+H_{21}\left[\begin{matrix}3d_{11} \\2d_{11}\end{matrix}\right]\\
&=\left[\begin{matrix}3d_{21}+3d_{22}+3d_{11}\\4d_{21}+5d_{11}\\5d_{11}\end{matrix}\right]
\end{align*}
With matrix inverse, both $t_1$ and $t_2$ can recover their desired symbols.

\end{example}

\section{Conclusion}\label{sec:conclusion}
In this paper, the capacity region of the two-user linear deterministic IC is derived, where the result is given in terms of the rank of the transition matrices. Our result is applicable to the scenarios where the channel matrices are correlated and/or rank deficient. To achieve the rate pairs in the capacity region, we combine the standard MIMO SVD technique and the idea of common-private rate splitting, based on which a simple linear precoder is developed. Moreover, we show that this linear deterministic IC can be used to model the double-unicast networks when random network coding is performed at all the intermediate nodes. Therefore, the capacity results derived are utilized to obtain an achievable region for the double-unicast networks, and it is proved that the region is strictly larger than the existing results in the literature.

However, there still exists a gap between our achievable region and the capacity of the double-unicast network as random linear network coding may be sub-optimal. One possible future work is to find a better network coding strategy by optimizing the transition matrices subject to the topology constraint, instead of using simple random network coding, such that the proposed achievable region is maximized.

\appendices
\section{A Useful Lemma}\label{A:usefulLemma}
\begin{lemma}\cite{Bernstein08}\label{lem:2matrixMultiply}
Let $A\in\mathbb{R}^{p\times l}$ and $B\in\mathbb{R}^{l\times k}$, $\mathrm{rank}(AB)=\mathrm{rank}(A)-\mathrm{dim}(\mathcal{N}(A^T)\cap \mathcal{R}(B^T))=\mathrm{rank}(B)-\mathrm{dim}(\mathcal{N}(A)\cap\mathcal{R}(B))$
\end{lemma}
\begin{proof}
Refer to pp.126 in \cite{Bernstein08}.
\end{proof}
\section{Proof of Lemma \ref{lem:MutliplyRank3}}\label{A:fullrank}
As the real field can be considered as $F_q$ with $q\rightarrow\infty$. Lemma \ref{lem:MutliplyRank3} is true if the following corollary holds.
\begin{corollary}\label{cor:FF}
Given  $A_1\in\mathbb{F}_q^{p\times l_1}$, $A_2\in\mathbb{F}_q^{p\times l_2}$ and $A_3\in\mathbb{F}_q^{p\times l_3}$, and let $E_1\in \mathbb{F}_q^{l_1\times k_1}$, $E_2 \in \mathbb{F}_q^{l_2\times k_2}$ and $E_3 \in \mathbb{F}_q^{l_3\times k_3}$ be  uniformly and independently generated, then $\mathrm{rank}([\begin{matrix}A_1E_1 & A_2E_2 & A_3E_3\end{matrix}])=k_1+k_2+k_3$ holds with probability approaching to 1 when $q\rightarrow \infty$, if the following conditions are satisfied:
\begin{itemize}
\item{$k_1\leq\mathrm{rank}(A_1)$}
\item{$k_2\leq\mathrm{rank}(A_2)$}
\item{$k_3\leq\mathrm{rank}(A_3)$}
\item{$k_1+k_2\leq\mathrm{rank}([\begin{matrix}A_1 & A_2\end{matrix}])$}
\item{$k_1+k_3\leq\mathrm{rank}([\begin{matrix}A_1 & A_3\end{matrix}])$}
\item{$k_2+k_3\leq\mathrm{rank}([\begin{matrix}A_2 & A_3\end{matrix}])$}
\item{$k_1+k_2+k_3\leq\mathrm{rank}([\begin{matrix}A_1 & A_2 & A_3\end{matrix}])$}
\end{itemize}
\end{corollary}
\begin{proof}
Before proving Corollary \ref{cor:FF}, we need to establish some useful facts.
\begin{fact}\cite{Blake06}\label{fact:countSubspace}
Let $\mathcal{V}_l(q)$ denote the vector space of dimension $l$ over the finite field $\mathbb{F}_q$, the number of distinct $k$-dimensional subspaces of $\mathcal{V}_l(q)$, a quantity denoted by $\left[\begin{matrix}l\\k\end{matrix}\right]_q$, is
\begin{equation*}
\left[\begin{matrix}l\\k\end{matrix}\right]_q=\prod_{i=0}^{k-1}\frac{(q^{l-i}-1)}{(q^{k-i}-1)}
\end{equation*}
\end{fact}

\begin{fact}\label{fact:MutliplyRank1}
Given a matrix $A\in\mathbb{F}_q^{p\times l}$, let $E$ be a random matrix uniformly generated from $\mathbb{F}_q^{l\times k}$, where $k\leq\mathrm{rank}(A)$, then $\mathrm{rank}(AE)=k$ holds with probability approaching to 1 when $q\rightarrow \infty$.
\end{fact}
\begin{proof}
According to lemma \ref{lem:2matrixMultiply}, $\mathrm{rank}(AE)=\mathrm{rank}(E)-\mathrm{dim}(\mathcal{R}(E)\cap\mathcal{N}(A))$. Since the elements in $E$ are uniformly and independently generated from the finite field $\mathbb{F}_q$, $\Pr\{\mathrm{rank}(E)=k\}=\prod_{j=0}^{k-1}\left(1-\frac{1}{q^{l-j}}\right)$ and it approaches to 1 as $q\rightarrow\infty$. Therefore, to prove Fact \ref{fact:MutliplyRank1}, we only need to show that $\lim_{q\rightarrow\infty}\Pr\{\mathrm{dim}(\mathcal R(E)\cap\mathcal N(A))=0\}=1$.

Provided with $\mathrm{rank}(E)=k$, $\mathcal{R}(E)$ can be viewed as a subspace of $\mathcal{V}_l(q)$ chosen uniformly from all the subspaces of dimension $k$. Note that the null space  of $A$, $\mathcal{N}(A)$, is another subspace of $\mathcal{V}_l(q)$  with dimension $l-\mathrm{rank}(A)$. Thus, the probability that $\mathcal{R}(E)$ has non-empty intersection with $\mathcal{N}(A)$ is given by
\begin{equation*}
\begin{aligned}
&\Pr\{\mathrm{dim}(\mathcal R(E)\cap\mathcal N(A))\geq 1\}
=\frac{\left[\begin{matrix}l-\mathrm{rank}(A)\\ 1\end{matrix}\right]_q\left[\begin{matrix}l-1\\k-1\end{matrix}\right]_q}
      {\left[\begin{matrix}l\\k\end{matrix}\right]_q}\\
&=\frac{(q^k-1)(q^{l-\mathrm{rank}(A)}-1)}{(q-1)(q^l-1)}
\end{aligned}
\end{equation*}
As a result, $\lim_{q\rightarrow\infty}\Pr\{\mathrm{dim}(\mathcal R(E)\cap\mathcal N(A))\geq 1\}=\lim_{q\rightarrow\infty}\frac{1}{q^{\mathrm{rank}(A)-k+1}}=0$ since $k\leq\mathrm{rank}(A)$. Therefore, $\lim_{q\rightarrow\infty}\Pr\{\mathrm{dim}(\mathcal R(E)\cap\mathcal N(A))=0\}=1$. This completes the proof of Fact \ref{fact:MutliplyRank1}.
\end{proof}

\begin{fact}\label{fact:MutliplyRank2}
Given matrices $A_1\in\mathbb{F}_q^{p\times l_1}$, $A_2\in\mathbb{F}_q^{p\times l_2}$, let $E_1$ and $E_2$ be two matrices uniformly and independently generated from $\mathbb{F}_q^{l_1\times k_1}$ and $\mathbb{F}_q^{l_2\times k_2}$, respectively. When $q\rightarrow\infty$, $\mathrm{rank}([\begin{matrix}A_1E_1 & A_2E_2\end{matrix}])=k_1+k_2$ holds with probability approaching to 1 if the following conditions are satisfied:
\begin{itemize}
\item{$k_1\leq\mathrm{rank}(A_1)$}
\item{$k_2\leq\mathrm{rank}(A_2)$}
\item{$k_1+k_2\leq\mathrm{rank}([\begin{matrix}A_1 & A_2\end{matrix}])$}
\end{itemize}
\end{fact}
\begin{proof}
For notational convenience, denote $\mathrm{rank}(A_1)$, $\mathrm{rank}(A_2)$ and $\mathrm{dim}(\mathcal R(A_1)\cap\mathcal R(A_2))$ by $r_1$, $r_2$ and $r_{12}$ respectively.
\begin{equation*}
\begin{aligned}
&\mathrm{rank}\left([\begin{matrix}A_1E_1 & A_2E_2\end{matrix}]\right)\\
&=\mathrm{rank}(A_1E_1)+\mathrm{rank}(A_2E_2)-\mathrm{dim}(\mathcal{R}(A_1E_1)\cap\mathcal{R}(A_2E_2))\\
&\overset{(a)}{=}k_1+k_2-\mathrm{dim}(\mathcal{R}(A_1E_1)\cap\mathcal{R}(A_2E_2))
\end{aligned}
\end{equation*}
where (a) follows from Fact \ref{fact:MutliplyRank1} as $k_1\leq r_1$ and $k_2\leq r_2$. Therefore, to prove Fact \ref{fact:MutliplyRank2}, we only need to show that, when $q\rightarrow\infty$, $\mathrm{dim}(\mathcal{R}(A_1E_1)\cap\mathcal{R}(A_2E_2))=0$ holds with overwhelm probability.  
Note that $\mathcal R(A_1E_1)$ can be viewed as a random subspace of $\mathcal R(A_1)$ with dimension $k_1$. Similarly, $\mathcal R(A_2E_2)$ can be viewed as a random subspace of $\mathcal R(A_2)$ with dimension $k_2$. Therefore, the probability that $\mathcal R(A_1E_1)$ has non-zero intersection with $\mathcal R(A_2E_2)$ is given by
\begin{align*}
&\Pr\{\mathrm{dim}(\mathcal{R}(A_1E_1)\cap\mathcal{R}(A_2E_2))\geq 1\}\\\displaybreak
&=\frac{\left[\begin{matrix}r_{12}\\ 1\end{matrix}\right]_q
           \left[\begin{matrix}r_1-1\\k_1-1\end{matrix}\right]_q
           \left[\begin{matrix}r_2-1\\k_2-1\end{matrix}\right]_q}
    {\left[\begin{matrix}r_1\\k_1\end{matrix}\right]_q
    \left[\begin{matrix}r_2\\k_2\end{matrix}\right]_q}\\
&=\frac{(q^{k_1}-1)(q^{k_2}-1)(q^{r_{12}}-1)}{(q^{r_1}-1)(q^{r_2}-1)(q-1)}
\end{align*}
 As a result, $\lim_{q\rightarrow\infty}\Pr\{\mathrm{dim}(\mathcal{R}(A_1E_1)\cap\mathcal{R}(A_2E_2))\geq 1\}=\lim_{q\rightarrow\infty}\frac{q^{k_1+k_2+r_{12}}}{q^{r_1+r_2+1}}=0$, where the last equality follows since $k_1+k_2\leq\mathrm{rank}([\begin{matrix}A_1&A_2\end{matrix}])=\mathrm{rank}(A_1)+\mathrm{rank}(A_2)-\mathrm{dim}(\mathcal{R}(A_1)\cap\mathcal{R}(A_2))=r_1+r_2-r_{12}$. Therefore, we can conclude that, with overwhelm probability, $\mathrm{dim}(\mathcal{R}(A_1E_1)\cap\mathcal{R}(A_2E_2))=0$. This completes the proof of Fact \ref{fact:MutliplyRank2}.
\end{proof}

Now, we are ready to prove Corollary \ref{cor:FF}. For brevity, denote $\mathrm{rank}(A_i)$ by $r_i$, $i\in\{1,2,3\}$, and $\mathrm{dim}(\mathcal R(A_i)\cap\mathcal R(A_j))$ by $r_{ij}$, $i,j\in\{1,2,3\},i\neq j$. Therefore, $\mathrm{rank}[\begin{matrix}A_i&A_j\end{matrix}]=r_i+r_j-r_{ij}$.


Assume that the given conditions in Corollary \ref{cor:FF} are satisfied, according to Fact \ref{fact:MutliplyRank1} and Fact \ref{fact:MutliplyRank2},   $\mathrm{dim}(\mathcal R(A_iE_i))=k_i,i\in\{1,2,3\}$ and $\mathrm{dim}(\mathcal{R}(A_iE_i)\cap\mathcal{R}(A_jE_j))=0$, $i,j\in\{1,2,3\},i\neq j$ holds with probability approaching to 1. Therefore, we have
\begin{align*}
&\mathrm{rank}\left(\left[\begin{matrix}A_1E_1 & A_2E_2 & A_3E_3\end{matrix}\right]\right)\\
=&\mathrm{rank}\left(\left[\begin{matrix}A_1E_1 & A_2E_2\end{matrix}\right]\right)
 +\mathrm{rank}\left(A_3E_3\right)\\
 &-\mathrm{dim}\left((\mathcal{R}(A_1E_1)+\mathcal{R}(A_2E_2))\cap\mathcal{R}(A_3E_3)\right)\\
=&\mathrm{rank}\left(A_1E_1\right)+\mathrm{rank}\left(A_2E_2\right)+\mathrm{rank}\left(A_3E_3\right)\\
  &-\mathrm{dim}\left(\mathcal{R}(A_1E_1)\cap\mathcal{R}(A_2E_2)\right)\\
  &-\mathrm{dim}\left((\mathcal{R}(A_1E_1)+\mathcal{R}(A_2E_2))\cap\mathcal{R}(A_3E_3)\right)\\
=&k_1+k_2+k_3-\mathrm{dim}\left((\mathcal{R}(A_1E_1)+\mathcal{R}(A_2E_2))\cap\mathcal{R}(A_3E_3)\right)
\end{align*}

Therefore, to prove Corollary \ref{cor:FF}, it is sufficient to show that $\mathrm{dim}\left((\mathcal{R}(A_1E_1)+\mathcal{R}(A_2E_2))\cap\mathcal{R}(A_3E_3)\right)=0$ holds with overwhelm probability. Denote the dimension of the intersection between $(\mathcal{R}(A_1E_1)+\mathcal{R}(A_2E_2))$ and $\mathcal{R}(A_3)$ by $\alpha$, i.e., $\alpha\triangleq\mathrm{dim}((\mathcal{R}(A_1E_1)+\mathcal{R}(A_2E_2))\cap\mathcal{R}(A_3))$.  The, the probability that $\mathcal{R}(A_3E_3)$ has non-empty intersection with $(\mathcal{R}(A_1E_1)+\mathcal{R}(A_2E_2))$ is given by
\begin{equation*}
\begin{aligned}
&\Pr\{\mathrm{dim}\left((\mathcal{R}(A_1E_1)+\mathcal{R}(A_2E_2))\cap\mathcal{R}(A_3E_3)\right)\geq 1\}\\
&=\frac{\left[\begin{matrix}\alpha\\1\end{matrix}\right]_q\left[\begin{matrix}r_3-1\\k_3-1\end{matrix}\right]_q}
      {\left[\begin{matrix}r_3\\k_3\end{matrix}\right]_q}=\frac{(q^{\alpha}-1)(q^{k_3}-1)}{(q^{r_3}-1)(q-1)}
\end{aligned}
\end{equation*}
$\lim_{q\rightarrow\infty}\Pr\{\mathrm{dim}\left((\mathcal{R}(A_1E_1)+\mathcal{R}(A_2E_2))\cap\mathcal{R}(A_3E_3)\right)\geq 1\}=0$ if $\alpha\leq r_3-k_3$. Therefore, in order to prove Corollary \ref{cor:FF}, it is sufficient to show that $\alpha=\mathrm{dim}((\mathcal{R}(A_1E_1)+\mathcal{R}(A_2E_2))\cap\mathcal{R}(A_3))\leq r_3-k_3$ holds with overwhelm probability.

\begin{align}
\alpha=&\mathrm{dim}((\mathcal{R}(A_1E_1)+\mathcal{R}(A_2E_2))\cap\mathcal{R}(A_3))\nonumber\\
=&\mathrm{rank}\left(\left[\begin{matrix}A_1E_1&A_2E_2\end{matrix}\right]\right)+\mathrm{rank}(A_3)\nonumber\\
  &-\mathrm{rank}\left(\left[\begin{matrix}A_1E_1&A_2E_2&A_3\end{matrix}\right]\right)\nonumber\\
=&k_1+k_2+r_3-\mathrm{rank}\left(\left[\begin{matrix}A_1E_1&A_2E_2&A_3\end{matrix}\right]\right)\nonumber\\
=&k_1+k_2+r_3-\mathrm{rank}\left(\left[\begin{matrix}A_1E_1&A_3\end{matrix}\right]\right)-\mathrm{rank}(A_2E_2)\nonumber\\
 &+\mathrm{dim}(\mathcal{R}(A_2E_2)\cap(\mathcal{R}(A_1E_1)+\mathcal{R}(A_3)))\nonumber\\
=&k_1+r_3-\mathrm{rank}\left(\left[\begin{matrix}A_1E_1&A_3\end{matrix}\right]\right)\nonumber\\
 &+\mathrm{dim}(\mathcal{R}(A_2E_2)\cap(\mathcal{R}(A_1E_1)+\mathcal{R}(A_3)))\nonumber\\
=&\mathrm{dim}(\mathcal{R}(A_1E_1)\cap\mathcal{R}(A_3))\nonumber\\
&+\mathrm{dim}(\mathcal{R}(A_2E_2)\cap(\mathcal{R}(A_1E_1)+\mathcal{R}(A_3)))\label{eq:pl4_1}
\end{align}
Next, we proceed to calculate $\mathrm{dim}(\mathcal{R}(A_1E_1)\cap\mathcal{R}(A_3))$ and $\mathrm{dim}(\mathcal{R}(A_2E_2)\cap(\mathcal{R}(A_1E_1)+\mathcal{R}(A_3)))$. For notational convenience, denote $\gamma=\mathrm{rank}(\left[\begin{matrix}A_1 & A_2 & A_3\end{matrix}\right])$ and define $(g)^{+}=\max\{g,0\}$.

\begin{fact}\label{fact:intersection}
$\lim_{q\rightarrow\infty}\Pr\{\mathrm{dim}(\mathcal R(A_3)\cap\mathcal{R}(A_1E_1))=(k_1-r_1+r_{13})^{+}\}=1$
\end{fact}
\begin{proof}
The probability that the dimension of the intersection between $\mathcal{R}(A_3)$ and $\mathcal{R}(A_1E_1)$ is greater or equal to $j$ is given by
\begin{equation*}
\begin{aligned}
&\Pr\{\mathrm{dim}(\mathcal R(A_3)\cap\mathcal{R}(A_1E_1))\geq j)\}
=\frac{\left[\begin{matrix}r_{13}\\j\end{matrix}\right]_q\left[\begin{matrix}r_1-j\\k_1-j\end{matrix}\right]_q}
      {\left[\begin{matrix}r_1\\k_1\end{matrix}\right]_q}\\
&=\frac{\prod_{i=0}^{j-1}{\frac{(q^{r_{13}-i}-1)}{(q^{j-i}-1)}}\prod_{i=0}^{k_1-j-1}{\frac{(q^{r_1-j-i}-1)}{(q^{k_1-j-i}-1)}}}
      {\prod_{i=0}^{k_1-1}{\frac{(q^{r_1-i}-1)}{(q^{k_1-i}-1)}}}
\end{aligned}
\end{equation*}
Therefore, $\lim_{q\rightarrow\infty}{\Pr\{\mathrm{dim}(\mathcal R(A_3)\cap\mathcal{R}(A_1E_1)\geq j)\}}=\frac{q^{(r_{13}-j)j}q^{(r_1-k_1)(k_1-j)}}{q^{(r_1-k_1)k_1}}=q^{-j^2+(k_1+r_{13}-r_1)j}$. If $j>(k_1-r_1+r_{13})^{+}$, $\lim_{q\rightarrow\infty}{\Pr\{\mathrm{dim}(\mathcal R(A_3)\cap\mathcal{R}(A_1E_1))\geq j)\}}=0$. For any $E_1\in F_q^{l_1\times k_1}$, it can be verified that $\mathrm{dim}(\mathcal R(A_3)\cap\mathcal{R}(A_1E_1))\geq (k_1-r_1+r_{13})^{+}$.  Thus, we can conclude that $\lim_{q\rightarrow\infty}\Pr\{\mathrm{dim}(\mathcal R(A_3)\cap\mathcal{R}(A_1E_1))=(k_1-r_1+r_{13})^{+}\}=1$.
\end{proof}

Following similar arguments as that in the proof of Fact \ref{fact:intersection}, with overwhelm probability, we have
\begin{equation}
\begin{aligned}
&\mathrm{dim}(\mathcal{R}(A_2E_2)\cap(\mathcal{R}(A_1E_1)+\mathcal{R}(A_3)))\\
=&(k_2-r_2+\mathrm{dim}(\mathcal{R}(A_2)\cap(\mathcal{R}(A_1E_1)+\mathcal{R}(A_3))))^{+}\label{eq:pl4_2}
\end{aligned}
\end{equation}
Furthermore,
\begin{align}
&\mathrm{dim}(\mathcal{R}(A_2)\cap(\mathcal{R}(A_1E_1)+\mathcal{R}(A_3)))\nonumber\\
=&\mathrm{rank}(A_2)+\mathrm{rank}\left(\left[\begin{matrix}A_1E_1&A_3\end{matrix}\right]\right)-\mathrm{rank}\left(\left[\begin{matrix}A_1E_1&A_2&A_3\end{matrix}\right]\right)\nonumber\\
=&r_2+\mathrm{rank}(A_1E_1)+\mathrm{rank}(A_3)-\mathrm{dim}(\mathcal{R}(A_3)\cap\mathcal{R}(A_1E_1))\nonumber\\
 &-\mathrm{rank}\left(\left[\begin{matrix}A_1E_1&A_2&A_3\end{matrix}\right]\right)\nonumber\\
=&r_2+k_1+r_3-(k_1-r_1+r_{13})^{+}-\mathrm{rank}\left(\left[\begin{matrix}A_1E_1&A_2&A_3\end{matrix}\right]\right)\nonumber\\
=&r_2+k_1+r_3-(k_1-r_1+r_{13})^{+}-\mathrm{rank}(A_1E_1)\nonumber\\
 &-\mathrm{rank}\left(\left[\begin{matrix}A_2&A_3\end{matrix}\right]\right)
  +\mathrm{dim}(\mathcal{R}(A_1E_1)\cap(\mathcal{R}(A_2)+\mathcal{R}(A_3)))\nonumber\\
=&r_{23}-(k_1-r_1+r_{13})^{+}+\mathrm{dim}(\mathcal{R}(A_1E_1)\cap(\mathcal{R}(A_2)+\mathcal{R}(A_3)))\nonumber\\
\overset{(b)}{=}&r_{23}-(k_1-r_1+r_{13})^{+}\nonumber\\
&+(k_1-r_1+\mathrm{dim}(\mathcal{R}(A_1)\cap(\mathcal{R}(A_2)+\mathcal{R}(A_3))))^{+}\nonumber\\
=&r_{23}-(k_1-r_1+r_{13})^{+}+(k_1+r_2+r_3-r_{23}-\gamma)^{+}\label{eq:pl4_3}
\end{align}
where (b) follows from the similar arguments as that in the proof of Fact \ref{fact:intersection}. By substituting \eqref{eq:pl4_3} into \eqref{eq:pl4_2}, we get
\begin{equation}
\begin{aligned}
&\mathrm{dim}(\mathcal{R}(A_2E_2)\cap(\mathcal{R}(A_1E_1)+\mathcal{R}(A_3)))=(k_2-r_2+r_{23}\\
&-(k_1-r_1+r_{13})^{+}+(k_1+r_2+r_3-r_{23}-\gamma)^{+})^{+}\label{eq:pl4_4}
\end{aligned}
\end{equation}
With Fact \ref{fact:intersection} and by substituting \eqref{eq:pl4_4} into \eqref{eq:pl4_1}, we have
\begin{align*}
\alpha=&\mathrm{dim}((\mathcal{R}(A_1E_1)+\mathcal{R}(A_2E_2))\cap\mathcal{R}(A_3))\\
=&(k_1-r_1+r_{13})^{+}+(k_2-r_2+r_{23}-(k_1-r_1+r_{13})^{+}\\
 &+(k_1+r_2+r_3-r_{23}-\gamma)^{+})^{+}
\end{align*}
The remaining task is to show that $\alpha\leq r_3-k_3$. Following cases are considered.\\
\textbf{Case I}: $k_1-r_1+r_{13}\leq 0$\\
In this case, $\alpha=(k_2-r_2+r_{23}+(k_1+r_2+r_3-r_{23}-\gamma)^{+})^{+}$. If $\alpha=0$, the result holds trivially as $k_3\leq r_3=\mathrm{rank}(A_3)$. Therefore, we only need to show that $k_2-r_2+r_{23}+(k_1+r_2+r_3-r_{23}-\gamma)^{+}\leq (r_3-k_3)$. If $k_1+r_2+r_3-r_{23}-\gamma\leq 0$, it reduces to $k_2-r_2+r_{23}\leq r_3-k_3$, which is equivalent to the given condition $k_2+k_3\leq r_2+r_3-r_{23}=\mathrm{rank}([\begin{matrix}A_2&A_3\end{matrix}])$. On the other hand, if $k_1+r_2+r_3-r_{23}-\gamma>0$, it reduces to $k_2+(k_1-\gamma)\leq -k_3$, which is equivalent to the given condition $k_1+k_2+k_3\leq\gamma=\mathrm{rank}([\begin{matrix}A_1&A_2&A_3\end{matrix}])$. Therefore, we conclude that $\alpha\leq(r_3-k_3)$ is true when $k_1-r_1+r_{13}\leq 0$. \\
\textbf{Case II}: $k_1-r_1+r_{13}>0$\\
In this case, to prove $\alpha\leq(r_3-k_3)$, it is sufficient to show that $(k_2-r_2+r_{23}-k_1+r_1-r_{13}+(k_1+r_2+r_3-r_{23}-\gamma)^{+})^{+}\leq (r_1+r_3-r_{13})-(k_1+k_3)$. As $k_1+k_3\leq r_1+r_3-r_{13}=\mathrm{rank}([\begin{matrix}A_1&A_3\end{matrix}])$, it holds trivially if the left hand side is reduced to zero. Therefore, we only need to show $(k_1+r_2+r_3-r_{23}-\gamma)^{+}\leq (r_2+r_3-r_{23})-(k_2+k_3)$. Again, it is true when the left hand side is zero as $k_2+k_3\leq\mathrm{rank}([\begin{matrix}A_2&A_3\end{matrix}])=r_2+r_3-r_{23}$. On the other hand, when the left hand side is greater than zero, it is further reduced to $k_1+k_2+k_3\leq \gamma$, which is exactly the same as the last condition given in Corollary \ref{cor:FF}.

This completes the proof of Corollary \ref{cor:FF} and hence Lemma \ref{lem:MutliplyRank3}.
\end{proof}

\section{Proof of Lemma \ref{lem:entropyMatrix}}\label{A:entropyMatrix}
This lemma is proved by similar arguments as that in section III of \cite{Prabhakaran07}. Denote $\mathrm{rank}(A)$, $\mathrm{rank}(B)$ and $\mathrm{rank}\left(\left[\begin{matrix}A \\ B \end{matrix}\right]\right)$ by $r_A$, $r_B$ and $r_{AB}$ respectively. Let $N_1$ be a matrix of size $l\times(l-r_{AB})$ whose column vectors form a basis for $\mathcal{N}\left(\left[\begin{matrix}A \\ B \end{matrix}\right]\right)$. Then we can find a matrix $N_2$ of size $l\times(r_{AB}-r_{B})$ such that the column vectors of $N_1$ and $N_2$ form a basis for $\mathcal{N}(B)$. Moreover, let $N_3\in\mathbb{R}^{l\times r_B}$, be the basis of $\mathcal {R}(B)$. Therefore, $\left[\begin{matrix}N_1 & N_2 & N_3\end{matrix}\right]$ spans the input space and we can find $\mathbf{x}'$ such that $\mathbf{x}=\left[\begin{matrix}N_1 & N_2 & N_3\end{matrix}\right]\mathbf{x}'$.
Thus,
\begin{align}
&H(A\mathbf{x}|B\mathbf{x})\\
&=H\left(A[\begin{matrix}N_1 & N_2 & N_3\end{matrix}]\mathbf{x}'\mid B[\begin{matrix}N_1 & N_2 & N_3\end{matrix}]\mathbf{x}'\right)\nonumber\\
&=H\left([\begin{matrix}0_{l\times(l-r_{AB})} & AN_2 & AN_3\end{matrix}]\mathbf{x}'\mid [\begin{matrix}0 & BN_3\end{matrix}]\mathbf{x}'\right)\label{eq:ConditionalEntropy}
\end{align}
Write $\mathbf{x}'=\left[\begin{matrix}\mathbf{x}_1' \\ \mathbf{x}_2' \\ \mathbf{x}_3'\end{matrix}\right]$. \eqref{eq:ConditionalEntropy} can be written as:
\begin{align*}
H(A\mathbf{x}\mid B\mathbf{x})&=H\left(AN_2\mathbf{x}_2'+AN_3\mathbf{x}_3'\mid BN_3\mathbf{x}_3'\right)
\end{align*}
As $\mathcal R(B)=\mathcal R(N_3)$, $\mathrm{dim}(\mathcal N(B)\cap\mathcal R(N_3))=0$. According to Lemma \ref{lem:2matrixMultiply}, $\mathrm{rank}(BN_3)=N_3=r_B$. Thus, $\mathbf{x}_3'$ is uniquely determined by $BN_3\mathbf{x}_3'$ and we have
\begin{align*}
H(A\mathbf{x}\mid B\mathbf{x})&=H\left(AN_2\mathbf{x}_2'+AN_3\mathbf{x}_3'\mid BN_3\mathbf{x}_3',\mathbf{x}_3'\right)\\
&=H\left(AN_2\mathbf{x}_2'\mid BN_3\mathbf{x}_3',\mathbf{x}_3'\right)\\
&\leq H(AN_2\mathbf{x}_2')\leq\mathrm{rank}(AN_2)\\
&\leq\mathrm{rank}(N_2)=r_{AB}-r_{B}
\end{align*}
Thus, the lemma follows.

\bibliographystyle{ieeetr}
\bibliography{ICReference}
\end{document}